\documentclass[11pt]{article}

\usepackage[latin1]{inputenc}

\usepackage{amstext,amsmath,amssymb,amsfonts,bbm}
\usepackage{hyperref}
\usepackage{authblk}
\usepackage{caption} 
\usepackage{epsfig} 
\usepackage{subfigure}
\usepackage{color} 
\usepackage{graphicx} 
\usepackage{braket} 
\usepackage{slashed} 
\usepackage{dsfont} 
\usepackage{amsthm}  
\usepackage{psfrag}  
\usepackage{float}  
\usepackage{cite}  

\usepackage[lmargin=75pt,rmargin=75pt,tmargin=60pt,bmargin=80pt]{geometry}

\newcommand{\be}{\begin{equation}}
\newcommand{\ee}{\end{equation}} 
\newcommand{\Tr}{{\rm Tr}}
 
\DeclareMathOperator{\im}{\mathrm{i}}

\newcommand{\cG}{\mathcal{G}}

\allowdisplaybreaks[4]

\newtheorem{proposition}{Proposition}
\newtheorem{definition}{Definition}
\newtheorem{lemma}{Lemma}

\theoremstyle{definition}

\renewcommand\theequation{\arabic{section}.\arabic{equation}}

\begin{document}
\title{\bf Notes on Tensor Models and Tensor Field Theories}

\author{Razvan Gurau}

\affil{\normalsize\it 
CPHT, CNRS, Ecole Polytechnique, Institut Polytechnique de Paris, Route de Saclay, 91128 PALAISEAU, France and Perimeter Institute for Theoretical Physics, 31 Caroline St. N, N2L 2Y5, Waterloo, ON,
Canada \authorcr email: rgurau@cpht.polytechnique.fr \authorcr \hfill}

\date{}

\maketitle

\hrule\bigskip

\begin{abstract}
Tensor models and tensor field theories admit a $1/N$ expansion and a melonic large $N$ limit which is simpler than the planar limit of random matrices  
and richer than the large $N$ limit of vector models. They provide examples 
of analytically tractable but non trivial strongly coupled quantum field theories and lead to a new class of conformal field theories.
We present a compact introduction to the topic, covering both some of the classical results in the field, like the details of the $1/N$ expansion, as well as recent developments. These notes are loosely bases on four lectures given at the Journ\'ees de physique math\'ematique Lyon 2019: Random tensors and SYK models.
\end{abstract}

\hrule\bigskip

\tableofcontents

\section{Introduction}

Quantum field theory (QFT) \cite{zinn1996quantum} accurately describes both the fundamental interactions in nature (like the electroweak \cite{Weinberg1,Salam1,Glashow1} and strong \cite{Fritzsch:2015jfa}  interactions) and condensed matter systems (like Ising spins \cite{Ising:1925em} or Fermi liquids \cite{Metzner:2011cw}). In particular it gives one of the most precise predictions in physics: the electron anomalous magnetic moment up to $10^{-10}$ relative error. Perhaps the most important lesson of quantum field theory is that physics changes with the energy scale, as captured by the renormalization group \cite{Wilson:1971dc,Wilson:1973jj,Polchinski:1983gv}.

By and large QFT has two regimes. On the one hand one has weakly coupled theories, like quantum electrodynamics. As the name suggests, these theories are almost free and the effect of interactions is well accounted for in perturbation theory. Making sense rigorously  of the perturbative expansion is quite non trivial \cite{GlimmJaffe}. However, given the right circumstances,  perturbative computations yield very accurate predictions. On the other hand one has strongly coupled theories, which are famously difficult to deal with. While successful numerical approaches have been developed, like lattice quantum chromodynamics, analytical results are much harder to come by. Two prominent strategies exist to deal with strongly coupled QFTs analytically.

One strategy is to consider theories endowed with constraining symmetries or integrability properties. For instance at a fixed point of the renormalization group a QFT becomes scale invariant and, more often than not, scale invariant theories are in fact conformally invariant. Conformal invariance is a very strong constraint \cite{DiFrancesco:1997nk} which allows one to bootstrap a plethora of results \cite{ElShowk:2012ht}. 

A second strategy consists in identifying new parameters, not related to the strength of the interaction, and attempt a perturbative study with respect to them.
An example of this is the so called ``large $N$'' field theory \cite{Moshe:2003xn}.
If the quantum field itself is a vector (or a matrix or a tensor) in some Hilbert space of dimension $N$, one can attempt to study the theory in a  $1/N$ expansion. This is a three step strategy:
\begin{itemize}
 \item[-] take $N$ large (infinite). In this limit the theory simplifies. The $1/N$ expansion is useful as long as the large $N$ limit is rich enough to be non trivial, but simple enough to be more manageable than the original theory.
 \item[-] compute the corrections to the large $N$ behavior order by order in $1/N$, at all orders.
 \item[-] resum the $1/N$ series or bound the rest. 
\end{itemize}

This strategy brings mixed results. With the exception of some models in dimension zero\footnote{By dimension we mean the dimension of the space or space-time on which the vector, matrix or tensor field theory is defined, that is QCD is a theory in $4$ dimensions \cite{'tHooft:1973jz}. In matrix and tensor models \cite{DiFrancesco:1993nw,RTM} there is a second notion of dimension: the Feynman diagrams have non trivial topology and encode combinatorial triangulations in various dimensions (two in the case of QCD). 
Thus \cite{DiFrancesco:1993nw} deals with matrix models in 0 dimensions, although they are relevant for two dimensional quantum gravity.}  \cite{thomasmatrix,expansioin6,Krajewski:2017thd} step 3 is almost never considered. Step 2 is again considered mostly in dimension zero \cite{double,double1,double2,DGR,GurSch,Bonzom:2014oua,Gurau:2015tua,Rivasseau:2017xbk}. 
Beyond the fact that the second and third step are seldom manageable, the two classical examples of a vector \cite{Berlin:1952zz,Stanley:1968gx}
or a matrix \cite{'tHooft:1973jz} field are somewhat disappointing already at step 1. Vector models are analytically tractable in the large $N$ limit and have plenty of applications \cite{Moshe:2003xn,Berges:2004yj}\footnote{They provide for instance explicit CFT duals to Fradkin--Vasiliev higher spin theories \cite{Klebanov:2002ja,Vasiliev:1999ba,Fradkin:1986qy,Fradkin:1987ks}.}. However, they are limited by the fact that at leading order in $1/N$ vector models do not give an anomalous scaling dimension for the field. Consequently one is stuck with either numerical studies \cite{Guida:1998bx} or  almost classical scaling. On the contrary, matrix models \cite{'tHooft:1973jz,Brezin:1977sv,DiFrancesco:1993nw} are too complicated to be resummed in the large $N$ (planar) limit, in more than zero dimension.

Tensor models \cite{RTM,Klebanov:2018fzb} give a new class of large $N$ field theories. They exhibit a \emph{melonic} large $N$ limit \cite{critical,RTM} which is different from both the vector and the matrix ones. Vector-tensor models and some regimes of matrix models also lead to a melonic limit \cite{Ferrari:2017ryl,Azeyanagi:2017drg,Ferrari:2017jgw,Azeyanagi:2017mre}. Unsurprisingly, the melonic limit is richer than the large $N$ limit of vectors. Surprisingly, although as algebraic objects tensors are more complicated than matrices, the melonic limit is simpler than the planar one. 

Tensor models have been extensively studied in zero dimensions (where they were originally introduced and studied as models of quantum gravity  \cite{ambj3dqg,sasa1,color,review,BenGeloun:2011rc,Samary:2014oya,Pascalie:2018nrs }) and in one dimension (e.g. \cite{Witten:2016iux,Gurau:2016lzk,Klebanov:2016xxf,Peng:2016mxj,Krishnan:2016bvg,Krishnan:2017lra,Bulycheva:2017ilt,Choudhury:2017tax,Pakrouski:2018jcc,Klebanov:2018nfp,Kim:2019upg,Carrozza:2018psc,Klebanov:2019jup}) because they provide an alternative to the Sachdev-Ye-Kitaev model \cite{Sachdev:1992fk,Kitaev,Maldacena:2016hyu, Polchinski:2016xgd,Gross:2016kjj,Gross:2017aos,Gross:2017hcz,Ferrari:2019ogc} without quenched disorder. 
Some small $N$ tensor models can also be solved analytically \cite{Krishnan:2017txw,Krishnan:2018hhu,Pakrouski:2018jcc,Klebanov:2018nfp}.
Higher dimensional tensor field theories have been more recently explored \cite{Giombi:2017dtl,Prakash:2017hwq,Benedetti:2017fmp,Giombi:2018qgp,Benedetti:2018ghn,Popov:2019nja}. At large $N$ and in the infrared these theories typically yield conformal field theories (CFTs) of a new kind which we call by extension melonic. 

 Melonic theories are an ideal compromise between solvability and richness:
 contrary to almost all the other examples of strongly interacting theories, 
 they can be treated analytically. To a large extent they can be studied disregarding their origin. This is reflected in the organization of these notes. In Section~\ref{sec:lect0} we present a brief overview of conformal field theories and the particular features of the melonic ones. In Section~\ref{sec:lect1} we present several models which become melonic in the large $N$ limit. Section~\ref{sec:2PI} presents an effective action formalism well adapted to tensor field theories and finally Section~\ref{sec:ren} presents in detail the renormalization group flow, fixed points and infrared melonic CFT in one model. The appendices collect some technical details.

 \bigskip
 
 \bigskip
 
\paragraph{\it Notation.} We work in Euclidean $\mathbb{R}^d$. We sometimes denote integrals over positions by $\int_x \equiv \int d^dx$ and integrals over momenta by  $\int_p \equiv \int \frac{d^dp}{(2\pi)^d}$. The Fourier transform is $  f(p) = \int_x  e^{\im p x} f(x)$ with inverse $f(x) = \int_p  e^{-\im px} f(p)$. Repeated indices are summed.

\newpage

\section{Melonic field theories}\label{sec:lect0}
\setcounter{equation}{0}

\bigskip
 
 Melonic conformal field theories are a new class  of analytically accessible CFTs. We first briefly review CFT in $d$ dimensions and then explain what makes melonic theories special. The reader can find plenty of references on conformal field theories. Here we present a brief digest of \cite{Liu:2018jhs,Simmons-Duffin:2017nub}. We use the notation in \cite{Liu:2018jhs}\footnote{This section is the author's synopsis of the four lectures given by V.  Rosenhaus at the Journ\'ees de physique math\'ematique Lyon 2019: Random tensors and SYK models. The author would like to take this opportunity to thank him for many clarifying discussions on the topic.}.
 
\subsection{A digest of conformal field theories}
 
In  Euclidean $\mathbb{R}^d$ with line element $dx^2 = \delta_{\mu\nu}\,dx^{\mu}dx^{\nu}$ conformal transformations $x \to x'(x)$ preserve the line element up to a local scale factor $dx'^2 = \Omega(x)^2 dx^2$. 
The infinitesimal\footnote{The finite transformations are 
translations $x'^{\mu} = x^{\mu}+ a^{\mu}$, rotations 
$x'^{\mu} = R^{\mu}_{\;\; \nu}x^{\nu}$, dilatations  $x'^{\mu} = \Lambda x^{\mu} $
and special conformal transformations $x'^{\mu}=( x^{\mu} + x^2 b^{\mu}) /(1 + 2b\cdot x + b^2 x^2)$.} conformal transformations $x'^{\mu} = x^{\mu}+v^{\mu}$, $\Omega =1 +\sigma$ are generated by:
\begin{equation}
 v_{\mu}(x) = a_{\mu} + \omega_{\mu\nu}x^{\nu} + \kappa x_{\mu} +
    b_{\mu} x^2 -2 \, x_{\mu} \; b\cdot x \;, \qquad \omega_{\mu\nu} = -\omega_{\nu\mu} \;,
    \qquad \sigma = \kappa -2 \, b\cdot x \; .
\end{equation}
The conformal group has $\binom{d+2}{2}$ generators and is locally isomorphic to $SO(d+1,1)$. A general conformal transformation is such that:
\begin{equation}\label{eq:conftr}
 \frac{\partial x'^{\mu}}{ \partial x^{\nu}} = \Omega(x) R^{\mu}_{\;\;\nu}(x) \;,\qquad
  \delta_{\mu\nu}R^{\mu}_{\;\;\sigma}(x)R^{\nu}_{\;\;\rho}(x) = \delta_{\sigma \rho} \;, 
  \qquad \Omega(x) =  
\left|  \frac{\partial x'^{\mu}}{ \partial x^{\nu}} \right|^{\frac{1}{d}} \;,
\end{equation}
and $(x'_1-x'_2)^2 =\Omega(x_1) \Omega(x_2) (x_1-x_2)^2$. Conformally invariant  cross ratios can be built starting with four positions 
$ ( x_{ij}^2 x_{kl}^2 )/( x_{ik}^{2} x_{jl}^2 )$ where $x_{ij} = x_i-x_j$.

The irreducible representations of $SO(d)$ are classified by the spin.
For bosonic fields the representation space of the spin $J$ representation consists in symmetric traceless tensors with $J$ indices. We denote multi indices by $\bar \mu = \mu_1\dots \mu_{J}$. The rotation $R^{\mu}_{\;\;\nu}$ is represented in the representation $J$ by the tensor product 
$R^{\bar \mu}_{\;\;\bar \nu } =  R^{\mu_1}_{\;\; \nu_1} \dots R^{\mu_J}_{\;\; \nu_J}$.
In a scalar theory for instance, a spin $J$ composite operator is 
$[ (\partial^2)^{n_1} \partial_{(\mu_1} \dots \partial_{\mu_i}\phi] [\partial_{\mu_{i+1}}\dots \partial_{\mu_J)} (\partial^2)^{n_2}\phi] -  {\rm traces}  $.

The primary operators $O_{\Delta,J}$ in a conformal field theory (CFT) are vectors in the spin representation $J$ and change under the conformal transformation in Eq.~\eqref{eq:conftr} 
as\footnote{In components, the infinitesimal transformation of primary fields is:
\[
   \delta O_{\Delta,J}^{\bar \alpha} = - v^{\mu} \partial_{\mu} 
   O_{\Delta,J}^{\bar\alpha}  
 - (\kappa - 2 b\cdot x) \Delta  O_{\Delta,J}^{\bar \alpha}  + 
 \frac{1}{2} \left[ \omega^{\mu\nu} - 2 (b^{\mu}x^{\nu} - b^{\nu} x^{\mu})\right]
  (s_{\mu\nu})^{\bar \alpha}_{\;\; \bar \beta} O_{\Delta,J}^{\bar \beta}    \;,
\]
where $s_{\mu\nu}$ denotes the spin matrices in the representation $J$. The spin matrices are the generators of the $so(d)$ Lie algebra $ \big[s_{\mu\nu}, s_{\rho\sigma} \big] = \delta_{\mu \rho} s_{\nu \sigma}
   -\delta_{\nu \rho} s_{\mu\sigma} - \delta_{\mu \sigma} s_{\nu \rho} + 
   \delta_{\nu \sigma} s_{\mu\rho}$. 
   In the vector representation for instance we have $(s_{\mu\nu})_{a b}
   = \delta_{\mu b} \delta_{\nu a}-\delta_{\nu b}\delta_{\mu a}$. 
}:
\begin{equation}
  O'^{\bar \mu}_{\Delta,J}(x')  = \Omega(x)^{-\Delta} R^{\bar \mu}_{\;\;\bar \nu}(x) 
  O_{\Delta,J}^{\bar \nu}(x)  \;,
\end{equation}
where $\Delta$ and $J$ are the scaling dimension and the spin of the operator 
and $R^{\bar \mu}_{\;\;\bar \nu}(x) $ is the matrix representing the rotation 
$R^{\mu}_{\;\;\nu}(x)$ in the spin representation $J$. 

The two and three point functions of primary operators are fixed by conformal invariance
 \cite{Dolan:2000ut,Simmons-Duffin:2017nub}. The two point function is non zero only for operators of the same dimension and spin:
\begin{equation}\label{eq:twodirect}
  \Braket{  O_{\Delta,J}^{\bar \mu}(x_1)  \;  O_{\Delta,J; \bar \nu}(x_2) } 
   =\frac{   I^{\mu_1}_{(\nu_1} (x_{12}) \dots I^{\mu_J}_{\nu_J)}(x_{12}) -{\rm traces} }
   {|x_{12}|^{2\Delta}} \;,
   \quad
  I^{\mu}_{\nu} (x)    
   =  \delta^{\mu}_{\nu} - 2\frac{ x^{\mu} x_{\nu} }{ |x|^2}  \;,
\end{equation}
while the three point function of two spin zero operators $\phi_{1}$ and $\phi_{2}$ with dimensions $\Delta_1$ and $\Delta_2$ and a spin $J$ operator $O_{\Delta,J}$ is: 
 \begin{equation}\label{eq:threedirect}
 \begin{split}
    \Braket{\phi_1(x_1) \phi_2(x_2) O_{\Delta,J;\bar \mu}(x_3) }
  & = C^{\Delta_1,\Delta_2}_{\Delta,J} \; 
   \frac{  Z_{\mu_1} \dots Z_{\mu_J}  - {\rm traces} }{|x_{12}|^{\Delta_1 + \Delta_2 - \Delta + J}
     |x_{13}|^{\Delta + \Delta_1 -\Delta_2 - J } 
     |x_{23}|^{\Delta + \Delta_2 - \Delta_1 - J} 
   } \;,
   \crcr
   Z_{\mu} &=   \left( \frac{ (x_{13})_{\mu}}{ |x_{13}|^2} 
     - \frac{(x_{23})_{\mu}}{|x_{23}|^2} \right)   \; ,
 \end{split}
 \end{equation}
where $C^{\Delta_1,\Delta_2}_{\Delta,J}$ are pure numbers.
 
The operator product expansion (OPE) in quantum field theory expresses the product of two operators at nearby points as a sum of local operators. For a scalar field theory this is written schematically as
$ \phi(x_1) \phi(x_2) \simeq \sum  C(x_{12}^2) \; x_{12}^{\mu_1}
  \dots x_{12}^{\mu_J} O_{\mu_1 \dots \mu_J} (x_2)$ for $x_1\sim x_2$. This equality should be interpreted in the weak sense, that is it is valid when inserted in arbitrary correlations. In a conformal field theory the OPE is strongly constrained by conformal 
  invariance\footnote{As an example, note that any operator in a CFT, primary or not, will change under global dilatations by a rescaling
$O'(\Lambda x) = \Lambda^{-\Delta_O} O(x)$ which fixes $C(x_{12}^2) \sim |x_{12}|^{  -  2\Delta_{\phi} + \Delta_{O} - J  } $.} and the sum restricts to primary operators\cite{Dolan:2000ut,Pappadopulo:2012jk}:
\begin{equation}\label{eq:OPEphiphi}
 \phi_1(x_1) \phi_2(x_2) = \sum_{\Delta,J} C^{\Delta_1,\Delta_2}_{\Delta,J}  \;
  P^{\Delta_1,\Delta_2;\bar \mu}_{\Delta,J}(x_{12},\partial_{x_2}) O_{\Delta, J;\bar \mu}(x_2) \;,
\end{equation}
with the OPE coefficients $ C^{\Delta_1,\Delta_2}_{\Delta,J} $ given by
Eq.~\eqref{eq:threedirect}  and $P^{\Delta_1,\Delta_2;\bar \mu}_{\Delta,J}(x_{12},\partial_{x_2})$ some universal differential operator 
fixed by conformal invariance which
captures the contribution of the primary $ O_{\Delta,J}$ and all its descendants. For instance the three point function of three spin zero operators is at the same time given by Eq.~\eqref{eq:threedirect} and by the OPE, hence: 
\begin{equation}
\frac{1}{|x_{12}|^{\Delta_1 + \Delta_2-\Delta} 
|x_{13}|^{\Delta_1 + \Delta - \Delta_2}|x_{23}|^{\Delta_2+\Delta - \Delta_1}} 
=P^{\Delta_1,\Delta_2}_{\Delta}(x_{12},\partial_{x_2}) \frac{1}{|x_{23}|^{2\Delta}} \;,
\end{equation}
where we omitted the spin index. The polynomial $ P^{\Delta_1,\Delta_2}_{\Delta}(x_{12},\partial_{x_2}) $ is 
obtained by substituting 
$x_{13} = x_{12}+x_{23}$ in the left hand side and Taylor expanding in $x_{12}$\footnote{At first orders we get $P^{\Delta_1,\Delta_2}_{\Delta}(x_{12},\partial_{x_2}) = |x_{12}|^{\Delta - \Delta_1 -\Delta_2} \left( 1 + \frac{\Delta + \Delta_1 - \Delta_2}{2\Delta} x_{12}^{\mu} \partial_{x_2^\mu} + \dots \right)$.}.

Arbitrary correlation functions in a conformal field theory can be computed by applying the OPE iteratively, therefore a CFT is completely specified by the list of primary operators and OPE  coefficients. We will now present a method for computing the dimensions of (some of) the physical primary operators and (some of) the OPE coefficients in a CFT.

Our starting point is a four point function. To simplify our life we consider correlations with four spin zero fields.
Applying the OPE twice in the channel $(12) (34)$ yields:
\begin{equation}
\Braket{\phi_1(x_1) \phi_2(x_2) \phi_3(x_3) \phi_4(x_4) }  
 = \sum_{\Delta,J}  C^{\Delta_1,\Delta_2}_{\Delta,J} 
 C^{\Delta_3,\Delta_4}_{\Delta,J} \;  G^{\Delta_i}_{\Delta,J}(x_i) \;,
\end{equation}
where the universal functions $G^{\Delta_i}_{\Delta,J}(x_i) $, the 
conformal blocks \cite{Costa:2011dw}, are known explicitly.
The four point function can  be re expressed in terms of conformal partial
 waves \cite{Liu:2018jhs,Simmons-Duffin:2017nub} as we now explain.

For any primary operator $O_{\Delta,J}$ we define its shadow $\tilde O_{\tilde \Delta,J}$ to be an operator with the same spin but with dimension 
$\tilde \Delta = d-\Delta$. Let us denote $\Braket{ \dots }^{\rm cs}$ the conformal structure of a correlation function, given by
Eq.~\eqref{eq:twodirect} 
and~\eqref{eq:threedirect} with the OPE coefficients set to $1$. The shadow coefficient \cite{Liu:2018jhs} of three operators $ S^{O_2O_3}_{O_1}  $ is defined by the equation:
\begin{equation}\label{eq:Sdef}
\begin{split}
&  \int d^dy \; 
\Braket{ \tilde O_{\tilde \Delta,J;\bar \nu}(x_1) 
\tilde O_{\tilde \Delta,J}^{\bar \mu}(y)    }^{\rm cs}  
 \Braket{ O_{\Delta,J;\bar \mu} (y)   O_2(x_2) O_3 (x_3) }^{\rm cs}  =\crcr
 & \qquad \qquad = 
S^{O_2O_3}_{O} \Braket{ \tilde O_{\tilde \Delta,J;\bar \nu}(x_1) O_2(x_2) O_3(x_3) }^{\rm cs} 
\; ,
\end{split}
\end{equation}
where from now on we assume we deal with real fields (otherwise the spin representation should be conjugated).  We have for instance \cite{Liu:2018jhs}:
\begin{equation}
\label{eq:Scoef}
\begin{split}
 S^{\Delta_1,\Delta_2}_{\Delta,J} & =   \pi^{\frac{d}{2}}
 \frac{ \Gamma\left(\Delta - \frac{d}{2} \right) \Gamma\left(\Delta  + J-1 \right)
 \Gamma\left( \frac{\tilde \Delta +\Delta_1 -\Delta_2+J}{2}\right)
 \Gamma\left(\frac{\tilde \Delta  - \Delta_1 + \Delta_2+J}{2} \right)
 }{\Gamma\left( \Delta -1\right)\Gamma\left( d-\Delta  + J \right)
 \Gamma\left( \frac{ \Delta +\Delta_1 -\Delta_2+J}{2} \right)
 \Gamma\left( \frac{ \Delta  - \Delta_1 + \Delta_2+J}{2} \right)} \; , \crcr 
 S^{\Delta_1,(\Delta,J) }_{\Delta_2} & = \pi^{\frac{d}{2}} 
  \frac{\Gamma\left(\Delta_2 - \frac{d}{2}\right) 
     \Gamma\left(\frac{\tilde \Delta_2 + \Delta_1 - \Delta + J}{2}\right) 
     \Gamma\left(\frac{\tilde \Delta_2 + \Delta - \Delta_1 + J}{2}\right) }
  {\Gamma\left(d-\Delta_2\right) 
  \Gamma\left(\frac{  \Delta_2 + \Delta_1 - \Delta + J}{2} \right) 
  \Gamma\left(\frac{  \Delta_2 + \Delta - \Delta_1 + J}{2} \right) } \;.
  \end{split}
\end{equation}

Let us define the conformal partial waves:
\begin{equation}
 \Psi^{\Delta_i}_{\Delta,J} = \int dx_0 
 \Braket{ \phi_1(x_1) \phi_2(x_2)  O_{\Delta,J}^{\bar \mu} (x_0) }^{\rm cs} 
 \Braket{  \tilde O_{\tilde \Delta,J;\bar \mu} (x_0) \phi_3(x_3) \phi_4(x_4) }^{\rm cs} \;.
\end{equation}
Using the conformal scaling in the first three point conformal structure, one can show after some effort \cite{Simmons-Duffin:2017nub} that the conformal partial wave is a sum of the conformal block $G^{\Delta_i}_{\Delta,J}(x_i) $ and its shadow block $G^{\Delta_i}_{\tilde \Delta,J}(x_i) $:
\begin{equation}
 \Psi^{\Delta_i}_{\Delta,J}   = \left( -\frac{1}{2}\right)^J S^{\Delta_3,\Delta_4}_{\tilde \Delta,J}
  G^{\Delta_i}_{\Delta,J}(x_i) +\left( -\frac{1}{2}\right)^J S^{\Delta_1,\Delta_2}_{\Delta,J}
  G^{\Delta_i}_{\tilde \Delta,J}(x_i)  \; .
\end{equation}

A complete set of partial waves $\Psi^{\Delta_i}_{\Delta,J} $ is obtained in $d > 1$ by choosing integer spin $J$ and the dimensions $\Delta = d/2 + \im r, \;r\ge 0$ (for $d=1$ one needs to add the discrete set $\Delta = 2n \;, n\ge 1$). These dimensions do not correspond to physical primary operators.
The functions are orthogonal \cite{Simmons-Duffin:2017nub}:
\begin{equation}
\begin{split}
  \left(  \Psi^{\Delta_{i}}_{\Delta,J}   ,  \Psi^{\tilde \Delta_{i}}_{\tilde \Delta',J'}  \right) &= \int \frac{  \prod_{i=1}^4 d^dx_i}{{\rm Vol}(SO(d+1,1))}   \; \Psi^{\Delta_{i}}_{\Delta, J} (x_i)    \Psi^{\tilde \Delta_{i}}_{\tilde \Delta',J'} (x_i) =  2\pi n_{\Delta,J}\delta(r-r') \delta_{J J'}  \;, \crcr
 n_{\Delta,J} &= \frac{  S^{\Delta_3,\Delta_4}_{\tilde \Delta,J} 
 S^{\tilde \Delta_3, \tilde \Delta_4}_{\Delta,J} {\rm Vol}(S^{d-2})}{ {\rm Vol}(SO(d-1))} \;\;
 \frac{ \pi (2J  + d-2)  \Gamma(J+1) \Gamma(J+d-2)}{2^{2J+d-2} 
 \Gamma\left( J + \frac{d}{2}\right)^2} \;,
 \end{split}
 \end{equation}
with $\Delta = d/2+\im r$, $\tilde \Delta' = d/2 -\im r'$ and $r,r'>0 $.

Now, let us consider a field theory (not necessarily conformal) for a scalar field $\phi$, such that the one point function is zero. The four point function is a sum of a disconnected contribution $(12)(34)$ and the part connecting $(12)$ and $(34)$:
\begin{equation}\label{eq:conndiscon}
 \Braket{\phi(x_1) \phi(x_2) \phi(x_3) \phi(x_4) }  
 = \Braket{\phi(x_1) \phi(x_2) } \Braket{\phi(x_3) \phi(x_4)}
  +  \Braket{\phi(x_1) \phi(x_2) \phi(x_3) \phi(x_4) }_{12\to 34}  \;.
\end{equation}
The correlation $\Braket{\dots}_{12\to 34}$ can be written in term of the irreducible four point kernel (see Section \ref{sec:2PI} for details).  Expressing the self-energy $\Sigma$ (that is the one particle irreducible two point function) in terms of the dressed two point function $G$, the irreducible four point kernel is the functional derivative of $\Sigma$ with respect to $G$:
\begin{equation}
 K(x_1,x_2 ; x_3, x_4) = \int d^dx_a d^dx_b \; G(x_{1a}) G(x_{2b}) \frac{\delta \Sigma(x_{34})}{\delta G(x_{ab})} \;,
\end{equation}
and the four point function connecting $(12)$ and $(34)$ is:
\begin{equation}\label{eq:foupoint1}
\begin{split}
&  \Braket{\phi(x_1) \phi(x_2) \phi(x_3) \phi(x_4) }_{12\to 34} = \crcr
& \qquad =
  \int d^dx_a d^dx_b \;
    \left(  \frac{1}{1-K} \right)(x_1,x_2 ;x_a x_b) 
    \bigg( G(x_{a3}) G(x_{b4}) + (a\leftrightarrow b) \bigg) \; .
\end{split}
\end{equation}

In a CFT in which the field $\phi$ is a primary operator with dimension 
$\Delta_{\phi}$, as the partial waves form a basis, Eq.~\eqref{eq:conndiscon} becomes:
\begin{equation}\label{eq:4expansion}
 \Braket{\phi(x_1) \phi(x_2) \phi(x_3) \phi(x_4) }  
 = \frac{1}{|x_{12}|^{2\Delta_{\phi}}} \; \frac{1}{|x_{34}|^{2\Delta_{\phi}}}
 + \sum_J \int_{\frac{d}{2}}^{\frac{d}{2} + \im \infty}
   \frac{d\Delta}{2\pi \im} \; \rho(\Delta,J) \Psi^{\Delta_{\phi}}_{\Delta,J}(x_i) \;,
\end{equation}
where the field is normalized so that the two point function is exactly the conformal structure. The disconnected term is the contribution to the OPE of the identity operator with dimension and spin $0$. All the other physical operators and the OPE coefficients are captured by the density $\rho(\Delta,J)$. This density can be computed by expanding 
Eq.~\eqref{eq:foupoint1} on partial waves. We first expand the rightmost term in Eq.~\eqref{eq:foupoint1}: 
\begin{equation}
\begin{split}
  \Braket{\phi(x_1) \phi(x_3)}
 \Braket{\phi(x_2) \phi(x_4)} + (1\leftrightarrow 2) = {\cal F}^0(x_i)   & = \sum_{J}
 \int_{\frac{d}{2}}^{\frac{d}{2}+ \im \infty} \frac{d\Delta}{2\pi \im} \;
   \rho^0(\Delta,J) \Psi^{\Delta_{\phi}}_{\Delta,J}(x_i) \;,  
\end{split}
\end{equation}
where $     \rho^0(\Delta,J)   =   \left( {\cal F}^0 ,  \Psi^{\tilde \Delta_{\phi}}_{\tilde \Delta,J} \right) / n_{\Delta,J} $. The first term in the scalar product is, substituting the partial wave (and denoting arguments as indices):
\begin{equation}
\begin{split}
 &  \int\frac{d^dx_i d^dx_0 }{  {\rm Vol}(SO(d+1,1)) }
   \Braket{\phi_{x_1} \phi_{x_3}}
 \Braket{\phi_{x_2} \phi_{x_4}} \Braket{ \tilde \phi_{x_1} 
 \tilde \phi_{x_2}  
 \tilde O_{\tilde \Delta,J}^{\bar \mu} (x_0) }^{\rm cs} 
 \Braket{  O_{ \Delta,J;\bar \mu} (x_0) \tilde \phi_{x_3} \tilde \phi_{x_4} }^{\rm cs}  \crcr
 & = S^{\tilde \Delta_{\phi}, (\Delta,J)}_{\tilde \Delta_{\phi}}
   S^{\Delta_{\phi}, (\Delta,J)}_{\tilde \Delta_{\phi}} 
    \int\frac{dx_1dx_2dx_0}{{\rm Vol}(SO(d+1,1)) } 
    \Braket{   \tilde \phi_{x_1} \tilde \phi_{x_2}  
 \tilde O_{\tilde \Delta,J}^{\bar \mu} (x_0) }^{\rm cs} 
 \Braket{  O_{ \Delta,J;\bar \mu} (x_0)  \phi_{x_1}   \phi_{x_2} }^{\rm cs} \;,
 \end{split}
\end{equation}
where we computed the integrals over $x_3$ and $x_4$ using Eq.~\eqref{eq:Sdef}.
The remaining integral is just a pure number \cite{Liu:2018jhs} which we 
denote $t_0$:
\begin{equation}
t_0 = \frac{1}{{\rm Vol}(SO(d-1))} \,\frac{\Gamma\left( \frac{d-2}{2} \right)
 \Gamma(J+d-2)}{ 2^J \Gamma(d-2) \Gamma\left( J + \frac{d-2}{2} \right) } \;.
\end{equation}
 Taking into account the symmetry properties of the conformal three point function we get:
\begin{equation}
 \rho^0 (\Delta,J)= \frac{1 + (-1)^J}{n_{\Delta, J}} t_0 \;  S^{\tilde \Delta_{\phi}, (\Delta,J)}_{\tilde \Delta_{\phi}}
   S^{\Delta_{\phi}, (\Delta,J)}_{\tilde \Delta_{\phi}}  \;.
\end{equation}

Now, due to conformal invariance, the irreducible four point kernel applied on a three point function must be proportional to the three point function:
\begin{equation}\label{eq:4pointCFT}
 \int dx^d_3 dx^d_4 \;  K(x_1,x_2 ; x_3, x_4) \Braket{\phi(x_3) \phi(x_4) 
    O_{\Delta,J}^{\bar \mu}(x) } = k(\Delta,J) 
 \Braket{\phi(x_1) \phi(x_2) 
    O_{\Delta,J}^{\bar \mu}(x) } \; ,
\end{equation}
therefore:
\begin{equation}
 \rho(\Delta,J) = \frac{1}{1-k(\Delta,J)} \rho^0(\Delta,J) \; .
\end{equation}
Putting everything together, inserting the partial wave in terms of the conformal blocks and noting that $\rho(\tilde \Delta,J) = \rho(\Delta,J)$ we get:
\begin{align}
& \Braket{\phi(x_1) \phi(x_2) \phi(x_3) \phi(x_4) }_{12\to 34}  = \\
&= \sum_J \int_{\frac{d}{2}-\im \infty}^{\frac{d}{2} + \im \infty}
   \frac{d\Delta}{2\pi \im} \; \frac{1}{1-k(\Delta,J)} 
    \frac{1 + (-1)^J}{n_{\Delta, J}} t_0 \;  S^{\tilde \Delta_{\phi}, (\Delta,J)}_{\tilde \Delta_{\phi}}
   S^{\Delta_{\phi}, (\Delta,J)}_{\tilde \Delta_{\phi}}
  \left(-\frac{1}{2} \right)^J S^{\Delta_{\phi}\Delta_{\phi}}_{\tilde \Delta,J} G^{\Delta_{\phi}}_{\Delta,J}(x_i) \;. \nonumber
\end{align}

In order to find the OPE coefficients and the dimension of the primaries, we close the integral contour on the right half complex plane. The integral then becomes a sum over the poles of the integrand. There are many poles: some come from the conformal block itself, some from the explicit $S$ factors and some from the $(1-k(\Delta,J))^{-1}$ factor. It turns out that some of the poles are spurious \cite{Simmons-Duffin:2017nub}, and only the poles of $1/( 1-k(\Delta,J) )$ are physical. We denote $\Delta_n$ the solutions of the equation $k(\Delta_, J)  = 1$. These are the dimensions of the physical primary operators present in the OPE of $\phi\phi$ in 
Eq.~\eqref{eq:OPEphiphi}, and: 
\begin{align}\label{eq:final}
& \Braket{\phi(x_1) \phi(x_2) \phi(x_3) \phi(x_4) }_{12\to 34} 
  = \sum_{J} \sum_{n} \left( C^{\Delta_{\phi}\Delta_{\phi}}_{\Delta_n,J} \right)^2
   G^{\Delta_{\phi}}_{\Delta_n,J}(x_i) \;, \\
&    \left( C^{\Delta_{\phi}\Delta_{\phi}}_{\Delta_n,J} \right)^2  =
     - {\rm Res} \left( \frac{1}{1-k(\Delta,J) } ; \Delta_n \right) 
    \frac{1 + (-1)^J}{n_{\Delta_n, J}} t_0 \;  S^{\tilde \Delta_{\phi}, (\Delta_n,J)}_{\tilde \Delta_{\phi}}
   S^{\Delta_{\phi}, (\Delta_n,J)}_{\tilde \Delta_{\phi}}
   \left( -\frac{1}{2} \right)^J S^{\Delta_{\phi}\Delta_{\phi}}_{\tilde \Delta_n,J}
     \; . \nonumber
\end{align}

This method for computing the dimensions of the physical primary operators and the OPE coefficients is completely general. However, it is of limited use in the most generic case because the four point kernel and consequently $k(\Delta,J)$ are complicated.
 
\subsection{The melonic truncation}
 
We now introduce a class of field theories which we call \emph{melonic}. In these theories one is able to close the equation \eqref{eq:final} and compute $k(\Delta,J)$ analytically.

Let us consider the simple example of a scalar field theory with a $q$--body interaction $\phi^q$ in zero dimensions. The ``field'' $\phi$ is just a real variable and the action and partition function write:
\begin{equation}\label{eq:firstex}
S = \frac{1}{2} \phi \, C^{-1}\phi + \frac{\lambda}{q!} \, \phi^q \;,\qquad
 {\cal Z} = \int [d\phi] \; e^{-S}
  \;,
\end{equation}
where $C>0$ is the covariance (propagator) and $\lambda$ the coupling. Of course in this case one can eliminate the covariance $C$ by a rescaling, but we refrain from doing this. 

The partition function and correlations can be evaluated by Taylor expanding in the coupling and computing the Gaussian integrals\footnote{For $\phi_A$ a vector in some vector space and $C_{AB}$ some non negative operator, the moments of the normalized Gaussian measure of covariance $C$ are computed by the Wick theorem:
\begin{align*}
\int   [d\phi] \; e^{-\frac{1}{2} \phi_A C^{-1}_{AB} \phi_B } \; \phi_{A_1} \dots \phi_{A_{2p}} = \sum_{{\rm pairings}\; \Pi} \prod_{\{i,j\} \in \Pi } C_{A_{i}A_{j}} \; ,
 \qquad [d\phi] \equiv (\det C)^{-1/2}  \prod_A d\phi_A  \;.
\end{align*}
}. This leads to the standard Feynman graph representation. The graphs have vertices with coordination $q$ and, for correlation functions, external points with coordination 1. The connected two point function of the model:
\begin{equation}
G = \Braket{\phi\phi}_c = \frac{1}{{\cal Z}} 
\int [d\phi] \; e^{-S } \; \phi\phi - \left( \frac{1}{{\cal Z}} 
\int [d\phi] \; e^{-S} \; \phi \right)^2 \;,
\end{equation}
is a sum over connect graphs with two external points. It obeys the Schwinger Dyson equation (SDE) depicted in Fig.~\ref{fig:SDE}:
\begin{equation}\label{eq:firstSDE}
 G^{-1} = C^{-1} - \Sigma \;,
\end{equation}
where the \emph{self energy} $\Sigma$ is the sum of amputated, one particle irreducible (1PI)
two point graphs.

\begin{figure}[htb]
\medskip
        \begin{center}
            \psfrag{g}{$G$}
    \psfrag{F}{$=\frac{1}{C^{-1} - \Sigma}$}
    \psfrag{c}{$C$}
    \psfrag{s}{$\Sigma$}
    \includegraphics[width=0.7\textwidth]{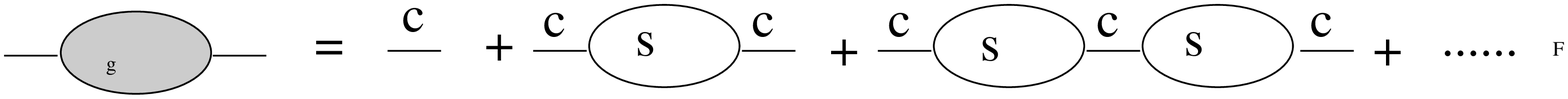}         
    \caption{The Schwinger Dyson equation.}  \label{fig:SDE}       
        \end{center}
\end{figure}

The SDE can be closed by re expressing the self energy back in terms of the two point function $G$. Usually this is not very useful as the self energy is a complicated sum over two particle irreducible graphs (more details on this in Section \ref{sec:2PI}).  

\begin{figure}[htb]
        \begin{center}
    \psfrag{g}{$G$}
    \psfrag{s}{$\Sigma$}
    \includegraphics[width=0.35\textwidth]{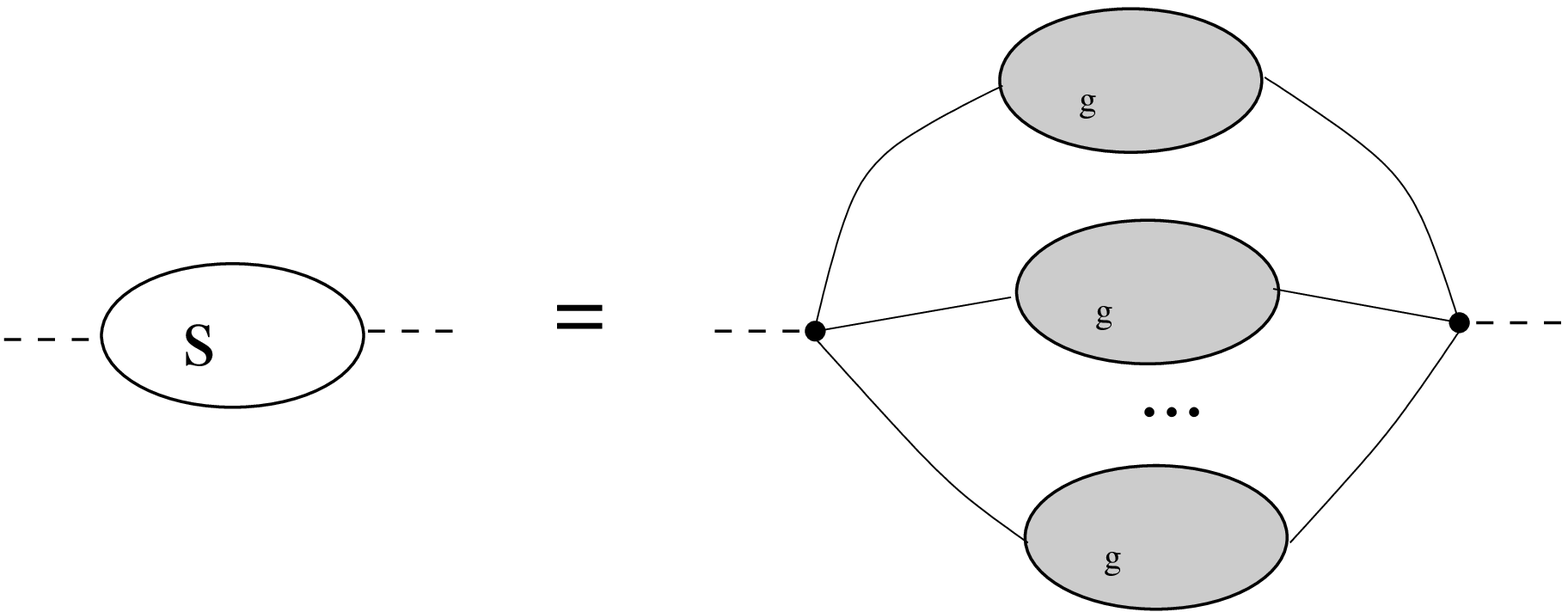}         
    \caption{The melonic truncation of the self energy.}  \label{fig:melon}       
        \end{center}
\end{figure}

The melonic truncation is a truncation of the self energy which leads to a non trivial but manageable SDE. It consists in restricting the self energy to the \emph{melon}  graph
depicted in Fig~.\ref{fig:melon} which is made of two vertices connected by $q-1$ parallel two point functions\footnote{The melonic truncation of the self energy defines melonic two point graphs. Vacuum melonic graphs are obtained by reconnecting the external edges of a (not necessarily one particle irreducible) two point graph into an edge. }. In zero dimensions the melonic truncation reads:
\begin{equation}\label{eq:firstmelon}
 \Sigma = \left[ \frac{ \lambda^2 }{(q-1)!} \right] \, G^{q-1} \; .
\end{equation}
We put the combinatorial factor in square brackets as it is the only thing which depends on the details of the model. We call a theory melonic if this truncation holds. 

In a zero dimensional melonic theory, combining  Eq.~\eqref{eq:firstSDE} and \eqref{eq:firstmelon} one can solve for the two point function:
\begin{equation}
  G  = C  + C  \left[ \frac{\lambda^2}{(q-1)!} \right]\, G^{q} \Rightarrow G = 
  \sum_{n\ge 0}  \frac{1}{qn+1} \binom{qn+1}{n} 
  \left(\left[ \frac{\lambda^2}{(q-1)!} \right]\, C^q\right)^n C \;.
\end{equation}

Now, let us go to higher dimensions. The SDE in a melonic theory in $d$ dimensions is:
\begin{equation}\label{eq:SDEx}
 (G^{-1})(x_1,x_2) = (C^{-1})(x_1,x_2) - \Sigma(x_1,x_2) \; \qquad 
 \Sigma(x) =  \left[ \frac{\lambda^2}{(q-1)!}  \right] \, \big[G(x) \big]^{q-1} \;,
\end{equation}
where $G^{-1}$ (and $C^{-1}$) denotes the operator inverse and we used translation invariance in the second equation. We now attempt to solve for the two point function self consistently. This is possible if one ignores the free covariance. Taking a conformal ansatz for the two point function:
\begin{equation}
 G(x_1-x_2) = b \, \frac{1}{|x_1-x_2|^{2\Delta_{\phi}}} \;, 
\end{equation}
and going to momentum space\footnote{Recall the Fourier transform:
\[
 \int_x \frac{e^{\im px}}{ x^{2a} } = 
 \frac{ \pi^{d/2}\Gamma\left( \frac{d}{2}-a\right)}{2^{2a-d}\Gamma(a)}\; \frac{1}{p^{d-2a}}
 \;.
\]} the SDE in a melonic theory is solved by:
\begin{equation}\label{eq:meloSDE}
 \Delta_{\phi} = \frac{d}{q} \;, \qquad
 \left[ \frac{\lambda^2}{(q-1)!} \right] b^q =  
 \frac{ \Gamma\left(\Delta_{\phi} \right) \Gamma \left( d-\Delta_{\phi} \right) }
 { \pi^d  (-1)\Gamma\left( \frac{d}{2} -\Delta_{\phi}  \right) 
 \Gamma\left( -\frac{d}{2} + \Delta_{\phi} \right) } \;.
\end{equation}

The attentive reader will note that this solution is only formal: the presence of an Euler $\Gamma$ function with a negative argument stems form the fact that the Fourier transform of the right hand side of Eq.~\eqref{eq:SDEx} is in fact divergent. We will be treating this equation rigorously in Section \ref{sec:ren}. Observe that the melonic truncation already expresses the self energy in terms of the two point function. The irreducible four point kernel is then readily obtained:
\begin{equation}
\begin{split}
 K(x_1,x_2 ; x_3, x_4) & = \int d^dx_a d^dx_b G(x_{1a}) G(x_{2b}) \frac{\delta \Sigma(x_{34})}{\delta G(x_{ab})} \crcr
   & = (q-1)  \left[ \frac{\lambda^2}{(q-1)!} \right] 
  G(x_{13}) G(x_{24}) G(x_{34})^{q-2}
 \; .
\end{split}
\end{equation}
In order to close Eq.~\eqref{eq:final} we need to determine $k(\Delta,J)$. The trick is to note that, as $\Delta_{\phi}=d/q$, we have:
\begin{equation}
 \frac{1}{|x_{34}|^{\Delta_{\phi}(q-2) } } \Braket{\phi(x_3) \phi(x_4) O_{\Delta,J}^{\bar \mu}(x_0)}^{\rm cs} =  \Braket{\tilde \phi(x_3) \tilde \phi(x_4) O_{\Delta,J}^{\bar \mu}(x_0)}^{\rm cs} \;,
\end{equation}
therefore:
\begin{align}
&   \int dx_3 dx_4 \; K(x_1,x_2; x_3x_4) 
 \Braket{\phi(x_3) \phi(x_4) O_{\Delta,J}^{\bar \mu}(x_0)}^{\rm cs}
 =  (q-1)  \left[ \frac{\lambda^2}{(q-1)!} \right] b^q 
 \\
& \qquad \qquad \qquad
  \int dx_3 dx_4 \Braket{\phi(x_1) \phi(x_3)}^{\rm cs }
  \Braket{\phi(x_2) \phi(x_4)}^{\rm cs } \Braket{\tilde \phi(x_3) \tilde \phi(x_4) O_{\Delta,J}^{\bar \mu}(x_0)}^{\rm cs}  \;. \nonumber
\end{align}
The integrals can now be computed using the shadow coefficients in Eq.~\eqref{eq:Sdef} and we get:
\begin{equation}
k(\Delta,J) =  (q-1)  \left[ \frac{\lambda^2}{(q-1)!} \right] b^q \;S^{\tilde \Delta_{\phi}, (\Delta,J)}_{\tilde \Delta_{\phi}}
   S^{\Delta_{\phi}, (\Delta,J)}_{\tilde \Delta_{\phi}}  \; ,
\end{equation}
which, with the help of Eq.~\eqref{eq:Scoef}, yields:
\begin{equation}\label{eq:puremelon}
  k(\Delta,J) =  (q-1)   \frac{\Gamma\left(  d- \Delta_{\phi}  \right) 
   \Gamma\left( \frac{d}{2}- \Delta_{\phi}\right) 
   \Gamma\left(\Delta_{\phi} -\frac{d}{2} + \frac{\Delta+J}{2} \right)
   \Gamma\left( \Delta_{\phi} - \frac{\Delta-J}{2}\right) 
   }
   {(-1) \Gamma\left( -\frac{d}{2} + \Delta_{\phi} \right) 
    \Gamma\left( \Delta_{\phi} \right)
     \Gamma\left( d - \Delta_{\phi} -\frac{\Delta-J}{2}\right)
      \Gamma\left( \frac{d}{2} -\Delta_{\phi} +\frac{\Delta +J}{2}\right)
    }  \;.
\end{equation}

 \newpage

 \section{The melonic limit as a large $N$ limit}\label{sec:lect1}
 \setcounter{equation}{0}
 
  \bigskip
  
  Melonic theories lead to analytically controlled CFTs in the infrared limit. However, in the previous section the melonic truncation appeared as a trick designed to produce a solvable model.

  The important question then becomes: is there any natural way to obtain a melonic limit in  a field theory? The answer to this question is yes: in the case of tensor field theories the melonic limit is naturally obtained at large $N$. In fact, when a random tensor is present, the large $N$ limit will often be melonic. In particular, as we will see below,
  models mixing vectors and tensors also fall in this class. 
  
  In this section we present three models which exhibit a melonic large $N$ limit. We only deal for now with the  combinatorial aspects of this limit and, in order to simplify the discussion, we will present the models in dimension zero. We will go back to field theories in the next section.

  The models we discuss here deal with non symmetric tensors. It should be mentioned that there exist models for symmetric tensors (in rank 3) for which the large $N$ limit has been proven to be melonic 
  \cite{Klebanov:2017nlk,Gurau:2017qya,Benedetti:2017qxl,
  Carrozza:2018ewt,Carrozza:2018psc}. However the proofs are significantly more involved for model with symmetries.

  \subsection{The colored tensor model}
   
This model is sometimes called the Gurau--Witten 
model\cite{color,review,Witten:2016iux}. It can be formulated for arbitrary rank tensors. For all the ranks it has a large $N$ limit dominated by melonic graphs \cite{critical}. The classification of graphs at any order in $1/N$ has been performed \cite{GurSch,Gurau:2016lzk}. 
 
Let us consider $D+1$ tensor fields of rank $D$. We denote the fields $T^i_{A^i}$ where $i=0,\dots D$ is the \emph{color} of $T^i$ and the multi index $A^i$ is $A^i = \{ a^i_j | j\neq i\}$. All the indices $a$ go from $1$ to $N$ and the tensors have no symmetry property. For example in rank $D=3$ the list of fields is:
\begin{equation}
 T^0_{a^0_1a^0_2a^0_3} \equiv T^0_{A^0} \; , \;\; 
 T^1_{a^1_0a^1_2a^1_3} \equiv T^1_{A^1} \;, \;\; 
 T^{2}_{a^2_0a^2_1a^2_2} \equiv T^2_{A^2} \;, \;\; 
 T^3_{a^3_0 a^3_1a^3_2} \equiv T^3_{A^3} \;.
\end{equation}
 Observe that the indices $a$ have two colors.  We denote $\delta_{A^iB^i} =\prod_{j\neq i} \delta_{a^i_j b^i_j}$.
 
 The model has a global symmetry group ${\cal O}(N)^{D(D+1)/2}$ consisting in an
orthogonal transformation $O^{(ij)} = O^{(ji)}\in {\cal O}(N)$ 
  for each couple of colors $(ij)$. Under the action of the symmetry group, \emph{both} indices $a^i_j$ and $a^j_i$ transform in the fundamental representation of $O^{(ij)}$. In detail, the global symmetry acts on $T^i$ as:
\begin{equation}
(T')^i_{B^i} = \prod_{j\neq i} O^{(ij)}_{b^i_j a^i_j} T^i_{A^i} \;, \qquad
  B^i = \{b^i_j | j\neq i \}\;, \;\; A^i = \{a^i_j | j\neq i\} \;.
\end{equation}
For example in rank $3$ we have $   O^{(10)} = O^{(01)}$  and so on and the fields transform as: 
\begin{equation}
(T')^0_{b^0_1 b^0_2b^0_3} = O^{(01)}_{b^0_1 a^0_1 } 
O^{(02)}_{b^0_2 a^0_2}O^{(03)}_{b^0_3a^0_3} T_{a^0_1a^0_2a^0_3}  \;,\;\;
 (T')^1_{b^1_0 b^1_2b^1_3} = O^{(10)}_{b^1_0 a^1_0 } 
 O^{(12)}_{b^1_2 a^1_2}O^{(13)}_{b^1_3a^1_3} T_{a^1_0a^1_2a^1_3} \;, \;\; {\rm etc.}\;. 
\end{equation}
The action and partition function of the model are:
   \begin{equation}
     S = \frac{1}{2}\sum_i T^i_{A^i} C^{-1} T^i_{A^i} + 
        \frac{ \lambda }{N^{ D(D-1) / 4 } } \prod_{i< j} \delta_{a^i_j a^{j}_i} \prod_i T^i_{A^i} \;,\qquad {\cal Z} = 
         \int[dT] \; e^{ -S } \;,
    \end{equation}
 where we have included a redundant covariance $C$ for the Gaussian part.
 Due to the global symmetry the two point functions of the model are diagonal both in the colors and in the indices:
   \begin{equation}\label{eq:2poincolor}
   \Braket{T^i_{A^i} T^j_{B^j}}_c = \delta^{ij}\delta_{A^iB^i} \, G \;, \qquad  
    G = C +  C \lambda \partial_{\lambda} (N^{-D} \ln {\cal Z} ) \; . 
    \end{equation}
 $G$ is obtained by taking a two point function, contracting its external indices respecting  the colors and dividing by $N^D$.

    The partition function and the correlations can be computed in the Feynman expansion.
    The  Feynman graphs are $D+1$ valent and the edges have a color $i=0,\dots D$. One can give a detailed, \emph{stranded}, representation of the Feynman graphs adapted to tracking the indices of the tensor. This is represented  in Fig.~\ref{fig:stranded} for $D=3$. Each tensor is represented as a half edge with $D$ strands, one for each one of its indices. $D+1$ half edges meet at a vertex and for every couple of half edges two strands (representing the indices $a^i_j$ and $a^j_i$) are joined. The edges transmit $D$ strands.
    
             \begin{figure}[htb]
        \begin{center}
    \includegraphics[width=0.4\textwidth]{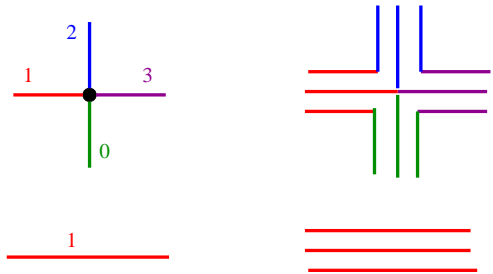}              
    \caption{The vertex and the propagator of the colored tensor model in rank $D=3$}\label{fig:stranded}
        \end{center}
     \end{figure}
     
    The vacuum Feynman graphs are edge $(D+1)$--colored graphs  \cite{color,review}. A graph $\cG$ has:
    \begin{itemize}
     \item[-] $V(\cG)$ vertices of coordination $D+1$. 
     \item[-] $\frac{D+1}{2}V(\cG)$ edges colored $0,1 \dots D$ such that at every vertex we have exactly one edge of each color.
     \item[-] $F(\cG)$ \emph{faces}, that is bi colored cycles.
    \end{itemize}
    
    The faces track the indices of the tensors: the indices are transmitted along the edges and turn around vertices, thus an index $a^i_j = a^{j}_i$ follows the face $(ij)$ and we get a free sum (hence a factor $N$) whenever the face closes. Open graphs arising in the Feynman expansion of correlations have additional  external points corresponding to the external field insertions and open strands connecting pairs of external indices. 
    
    In order to compute $G$ in Eq.~\eqref{eq:2poincolor} we use a trick (depicted in Fig.~\ref{fig:trick}): we contract the external indices of a two point function respecting the colors and divide by $N^D$.
    We denote $\mathfrak{G}$ the set of rooted\footnote{A rooted graph is a graph with one edge (the root) marked by an arrow. For colored graphs we fix the color of the root edge.}, connected edge $(D+1)$--colored graphs and we get:
   \begin{equation}
    G = \frac{1}{N^D} \delta_{A^iB^i} \Braket{T^i_{A^i} T^i_{B^i} }_c = 
    \sum_{\cG \in \mathfrak{G}} (-\lambda)^{V(\cG)} C^{\frac{D+1}{2} V(\cG)+1}
      N^{-D -\frac{D(D-1)}{4} V(\cG) + F(\cG)} \; ,
   \end{equation}
  where the root edge represents the external contraction $ \delta_{A^iB^i} $. Remarkably, every rooted graph with unlabeled vertices has a combinatorial factor 1. Observe that $\mathfrak{G}$ contains a graph with no vertices (on the left in Fig.~\ref{fig:trick}). It corresponds to the Gaussian pairing of the two external tensors and brings a covariance $C$. 
  
    \begin{figure}[htb]
        \begin{center}
    \includegraphics[width=0.3\textwidth]{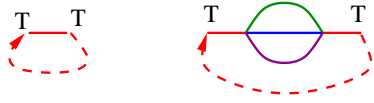}              
    \caption{Rooted graphs.}\label{fig:trick}
        \end{center}
     \end{figure}
  
  The crucial fact is that the numbers of faces and vertices of a connected graph $\cG$ are related \cite{expansion1,expansion2,expansion3} by the following relation (see Proposition\ref{prop:facecount}, Appendix \ref{app:deg}):
  \begin{equation}\label{eq:counting}
         F(\cG) = D + \frac{D(D-1)}{4} V (\cG) -\hat \omega(\cG) \; , \qquad \hat \omega(\cG) \ge 0 \;,
  \end{equation}
  where the \emph{reduced degree} $\hat \omega(\cG) $ of a graph $\cG$ is a non negative half integer.
  The properties of the degree are discussed in detail in Appendix \ref{app:deg}. The two point function (and all the other correlation functions) admits a $1/N$ expansion indexed by the degree:
 \begin{equation}
  G = \sum_{\hat \omega \in \mathbb{N}/2} N^{-\hat\omega}\sum_{\cG \in    \mathfrak{G}  }^{\hat \omega(\cG) = \hat \omega } (-\lambda)^{V(\cG)} C^{\frac{D+1}{2} V(\cG)+1} \;. 
 \end{equation}
 
At leading order one obtains only the graphs with $\hat\omega(\cG) = 0$. The graphs of degree zero are the melonic graphs \cite{critical}, see Proposition \ref{prop:melon}, Appendix \ref{app:deg}. Opening a rooted melonic vacuum graph at the root one obtains a melonic two point graph and at leading order:
\begin{equation}
 (G^{LO})^{-1} = C^{-1} - \Sigma^{LO} \;,\qquad \Sigma^{LO} = \lambda^2 (G^{LO})^{D} 
 \; .
\end{equation}

\subsection{The colored tensor--vector model}

This model is the zero dimensional counterpart of the Sachdev--Ye--Kitaev model  \cite{Sachdev:1992fk,Kitaev,Maldacena:2016hyu,Polchinski:2016xgd,Gross:2016kjj,Gurau:2017xhf}. It comes in two flavors, quenched and annealed, which coincide at the first few orders. Although more involved, the sub leading corrections have been classified also in this case \cite{Fusy:2018xax}. 

 The model consists in $D$ vectors $\psi^i_{a_i}$ (distinguished by the color $i$) coupled by a 
 random coupling $T_{a_1,\dots a_D}$. The random coupling is a rank $D$ tensor with no symmetries distributed on a Gaussian. The action of the model is:
     \begin{equation}
      S =  \frac{1}{2} \sum_i \psi^i_{a_i}C^{-1}\psi^i_{a_i} + \lambda \; T_{a_1\dots a_D} \prod_{i} \psi^i_{a_i} \;.
     \end{equation}

 We denote $T\cdot T \equiv T_{a_1\dots a_D} T_{a_1\dots a_D}$ and $[dT] = \prod_{a_1\dots a_D}
 N^{(D-1)/2} dT_{a_1\dots a_D}$. One can either take  the quenched or the annealed averages over the random couplings. Consequently one has the quenched and the annealed free energies, the quenched and the annealed two point functions, and so on. Denoting  ${\cal Z}(T)  = \int [d\psi] \; \exp\{ - S \} $ we have:
     \begin{align}
     & {\cal F}^{(q)}  = \frac{1}{N} \int [dT] \; e^{-\frac{ N^{D-1} }{2  } T\cdot T  }\ln \big( {\cal Z}(T) \big) \;, \qquad
     {\cal F}^{(a)}  = \frac{1}{N} \ln \left( \int [dT] \; e^{-\frac{N^{D-1}}{2} T\cdot T   }  {\cal Z}(T) \right) \;,    \crcr
 & \Braket{\psi^i_{a_i} \psi^j_{a_j} }_c^{(q)}  = 
  \int[dT] \; e^{-\frac{ N^{D-1} }{2  } T\cdot T } 
   \frac{1}{  {\cal Z}(T)  } \int [d\psi] \; e^{-S} \psi^i_{a_i} \psi^j_{a_j}    \;, \crcr 
 &   \Braket{\psi^i_{a_i} \psi^j_{a_j} }_c^{(a)}  = 
      \frac{1}{ 
      \int[dT] \; e^{-\frac{ N^{D-1} }{2  } T\cdot T   }
      {\cal Z}(T) }
  \int[dT] \; e^{-\frac{ N^{D-1} }{2  } T \cdot T  } 
  \int [d\psi] \; e^{-S} \psi^i_{a_i} \psi^j_{a_j}  \; ,
 \end{align}
The two point functions (both the quenched and the annealed one) are again diagonal in the colors and in the indices:
\begin{equation}
 \Braket{\psi^i_{a_i} \psi^j_{a_j} }_c^{(q),(a)}  = \delta^{ij} \delta_{a_ia_j} G^{(q),(a)} \;.
\end{equation}

The Feynman graphs are still edge $(D+1)$--colored graphs: $0$ is the color of the tensor (disorder) averages and $1,\dots D$ are the colors of the vector contractions. 
One can give a stranded represent in which the tensor has $D$ strands and the vectors only one strand as depicted in Fig.~\ref{fig:SYKstranded} on the right.
        
   \begin{figure}[htb]
        \begin{center}
    \includegraphics[width=0.5\textwidth]{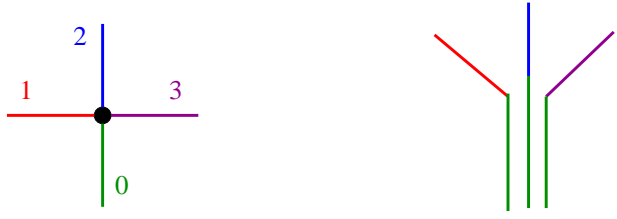}              
    \caption{Colored and stranded representation of the vertex of the tensor--vector model in rank $D=3$.}\label{fig:SYKstranded}
        \end{center}
     \end{figure}
     
     From the point of view of the index contractions the edge $0$ is very different from the others: the edges $1,\dots D$ carry an index each, while the edge $0$ carries $D$ indices. By the same trick as before the two point functions can be expressed as sums over rooted (the root has  color $i\neq0$) connected edge $(D+1)$--colored graphs:
\begin{equation}
\begin{split}
   G^{(q),(a)}  =& \frac{1}{N} \delta_{a_ia_j} 
  \Braket{\psi^i_{a_i} \psi^j_{a_j} }_c^{(q),(a)} \crcr
 &\qquad = \sum_{\cG \in \mathfrak{G}^{(q),(a)}} \ (-\lambda)^{ V(\cG) }
    C^{\frac{D}{2}V(\cG) +1 }   N^{-1-(D-1) V(\cG) + \sum_i F^{(0i)}(\cG)} \;,
\end{split}
\end{equation}
where $F^{(0i)}$ denotes the number of faces with colors $(0i)$ of $\cG$.
         
Contrary to the colored tensor model, we now get a free sum (hence a factor $N$) only for the faces involving the color $0$: it is quite  clear in the stranded representation of Fig.~\ref{fig:SYKstranded} that there is no index going from the edge $1$ to the edge $2$. Therefore the faces $(12)$, which are the bi colored cycles made by edges of colors $1$ and $2$ do not contribute to the amplitude. It is convenient to add and subtract the missing faces, $(ij)$ with $i,j\neq 0$. The number of faces which involve the color zero is the total number of faces minus the number of faces which do not involve the color $0$:
\begin{equation}
\sum_i F^{(0i)}(\cG)  = \sum_{0\le i<j \le D} F^{(ij)}(\cG) - \sum_{1\le i<j \le D} F^{(ij)}(\cG) \;.
\end{equation}

The only difference between the quenched and annealed models is the class of graphs over which we sum. In the annealed case we sum over all the (rooted) connected edge $(D+1)$--colored graphs $\mathfrak{G}^{(a)} \equiv \mathfrak{G}$. For the quenched case, we sum only over the (rooted) connected edge $(D+1)$--colored graphs which remain connected after deleting all the edges of color $0$. We denote the set of graphs with this property  $\mathfrak{G}^{(q)}$.
       
Let $\cG$ be a connected edge $(D+1)$--colored graph, and let us denote $\cG^0$ the edge $D$--colored graph obtained from $\cG$ by deleting the edges of color 0 (which correspond to the disorder averages). In general $\cG^0$ can be disconnected and we denote $C(\cG^0)\ge 1$ the number of connected components of $\cG^0$. As $\cG^0$ is an edge $D$--colored graph it has a reduced degree $\hat\omega(\cG^0)$ and the total number of its faces is given by Eq.~\eqref{eq:counting} with $D$ shifted to $D-1$.
We define the SYK degree of $\cG$ to be $\hat \Omega(\cG)  = \hat \omega(\cG) - \hat \omega(\cG^0)$.
This number is a half integer which obeys the bounds (see Proposition \ref{prop:SYK} in Appendix \ref{app:deg}):
\begin{equation}
\frac{1}{D} \hat \omega(\cG) \le   \Omega(\cG)    \le \hat \omega(\cG) \; ,
\end{equation}
hence in particular it is non negative. A straightforward computation yields:
\begin{equation}
G^{(q),(a)}  = \sum_{\cG \in \mathfrak{G}^{(q),(a)}} \ (-\lambda)^{ V(\cG) }
C^{\frac{D}{2}V(\cG) +1 }  N^{ - \big[ C(\cG^0)-1\big] (D-1) -  \Omega(\cG) } \; .
\end{equation}
For the quenched model we always have $C(\cG^0) = 1 $, but for the annealed model $C(\cG^0) \ge 1$. However, in the large $N$ limit, both the quenched and the annealed models are dominated by graphs with $C(\cG^0) = 1 $ and $\Omega(\cG) =0 $. The quenched and the annealed models coincide up to the order $N^{-(D-1)}$. If one uses the replica trick to compute the quenched averages, the departure between the quenched and the annealed models signals a replica symmetry breaking.
     
The SYK degree $\Omega(\cG)$ is non negative and is zero for a connected graph $\cG$ if and only if $\cG$ is melonic. If $\cG$ is melonic, then $\cG^0$ is a union of melonic graphs. At leading order $\cG^0$ is furthermore connected, hence it consists in exactly one melonic graph\footnote{Observe that in this case one can uniquely reconstruct $\cG$ starting from $\cG^0$.}. Therefore at leading order we get:
\begin{equation}\label{eq:SYKSDE}
 (G^{LO})^{-1} = C^{-1} - \Sigma^{LO} \;,\qquad \Sigma^{LO} = \lambda^2 (G^{LO})^{D-1} \;.
\end{equation}
       
       The annealed model (and consequently the quenched model at first orders) can be simplified by introducing ``bi local'' fields integrating over the disorder:
\begin{align}
& \int [dT] [d\psi] \; e^{-S} =\crcr
& \qquad= \int[d\psi] \; e^{-\frac{1}{2} \sum_i \psi^i_{a_i}C^{-1} \psi^i_{a_i} + \frac{\lambda^2}{2N^{D-1}} \prod_{i} \psi^i_{a_i}\psi^i_{a_i}} \int [dG^i][d\Sigma^i]  e^{ - \frac{1}{2} \sum_i \Sigma^i (NG^i - \psi^i_{a_i}\psi^i_{a_i} )} \crcr
& \qquad= \int [dG^i][d\Sigma^i] [d\psi]\; e^{ -\frac{1}{2} \sum_i  \psi_{a_i} (C^{-1} -\Sigma^i) \psi_{a_i} -N \big( \frac{1}{2} \sum_i G^i \Sigma^i   - \frac{\lambda^2}{2} \prod_i G^i \big)} \crcr
& \qquad = \int [dG^i][d\Sigma^i] \; e^{  -N \big( \pm \frac{1}{2} \sum_i \Tr\ln(1 - C \Sigma^i) + \frac{1}{2}\sum_i G^i \Sigma^i   - \frac{\lambda^2}{2} \prod_i G^i \big)} \;,
\end{align}
   where the $-$ sign is obtained for fermionic vectors (which requires even rank $D$), the $+$ sign for bosonic ones, and we took into account that the Gaussian integral over the vector fields is normalized. The advantage of this representation is that $N$ is an overall scaling and the $1/N$ expansion is a standard saddle point approximation. The saddle point equations write (using the fact that the saddle is color symmetric):
   \begin{equation}
    \Sigma - \lambda^2 G^{D-1} =0 \; ,\qquad G - \frac{1}{C^{-1}-\Sigma} =0 \;,
   \end{equation}
   which reproduce the SDE in the melonic limit \eqref{eq:SYKSDE}. The $1/N$ expansion is obtained by computing the saddle point corrections. However, we stress that this gives the $1/N$ expansion of the annealed model, hence fails to reproduce the one of the quenched case starting with the order $N^{-(D-1)}$. 

   \subsection{The ${\cal O}(N)^3$ model}
   
This model is sometimes called the Carrozza--Tanasa--Klebanov--Tarnopolsky model \cite{Carrozza:2015adg,Klebanov:2016xxf,Giombi:2017dtl}. As the name suggests, the model is defined only for rank 3  tensors. Its interest resides in the fact that it includes all the radiative corrections for quartic interactions, hence a field theory built on it is well adapted to a renormalization group study. At leading order the model is dominated by a \emph{melon-tadpole} graphs, a slight generalization of melons. The first sub leading orders of this model are understood \cite{Benedetti:2018goh,Carrozza:2015adg}.
 
 The field of the model is a rank 3 non symmetric tensor $\phi_A = \phi_{a_1a_2a_3}$
 transforming in the three fundamental representation of ${\cal O}(N) \otimes {\cal O}(N) \otimes {\cal O}(N)$, that is under a change of basis each index turns with its own orthogonal transformation:
 \begin{equation}(\phi')_{b_1b_2b_3} = O^{(1)}_{b_1a_1}O^{(2)}_{b_2a_2} O^{(3)}_{b_3a_3} \phi_{a_1a_2a_3}  \; .\end{equation}
One can consider the slightly more general case of a ${\cal O}(N_1) \otimes {\cal O}(N_2) \otimes {\cal O}(N_3) $ symmetry, but we will refrain from doing this here. We  denote by capital letters triples of indices: $A = \{ a_1,a_2,a_3\}$ and so on and we define three patters of contraction of indices among four tensors:
     \begin{align}
       &  \delta^t_{ABCD} = ( \delta_{a_1b_1} \delta_{c_1d_1} ) (\delta_{a_2c_2} \delta_{b_2d_2}) (\delta_{a_3d_3 } \delta_{b_3c_3}) \;,   \crcr
       & \delta^p_{AB;CD} = \frac{1}{3} \sum_{j} \delta_{a_jc_j} \delta_{b_jd_j}
    \left( \prod_{i\neq j} \delta_{a_ib_i} \delta_{c_id_i} \right)
    \; \qquad \delta^d_{AB;CD} = \prod_{i} \delta_{a_ib_i} \delta_{c_id_i} \;.
     \end{align}
The first pattern $ \delta^t_{ABCD}$ is called the \emph{tetrahedral} pattern, the second one $\delta^p_{AB;CD} $ the \emph{pillow} and the third one the \emph{double trace}. The action of the model is:
     \begin{equation}
      S = \frac{1}{2} \phi_{A} C^{-1}\phi_A +\left(  \frac{\lambda}{4N^{3/2}} 
      \delta^t_{ABCD} + \frac{\lambda_p}{4N^2} \delta^p_{AB;CD} +
        \frac{\lambda_d}{4N^3} \right)
         \phi_{A} \phi_{B} \phi_{C} 
         \phi_{ D } \; ,
     \end{equation}
and the two point function is diagonal in the tensor indices $\Braket{ \phi_A \phi_B } = \delta_{AB} G $.

Again one can compute the partition function and correlations in a Feynman expansion.
The Feynman graphs of the CTKT model are stranded graph  made of stranded vertices connected by stranded edges, as depicted in 
Fig.~\ref{fig:textCTKTvertices}.
The strands are associated to the indices of the tensors.
All the vertices are four valent and the edges have three strands.
The strands have a color and close into the faces of the graph. The faces correspond, again, to free sums over indices.      
From left to right in Fig.~\ref{fig:textCTKTvertices} we represented the  tetrahedral, the pillow and the double trace vertex. There are three kinds of pillow vertices, as a function of the special color which is transmitted from on pair of half edges to the other. 
      \begin{figure}[htb]
        \begin{center}
    \includegraphics[width=0.4\textwidth]{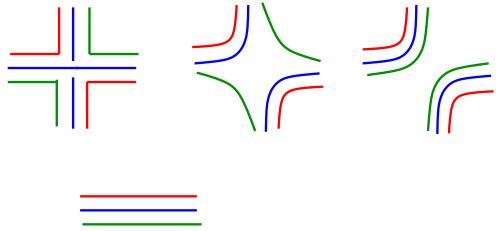}     
    \caption{Vertices and edges of the CTKT model}  \label{fig:textCTKTvertices}       
        \end{center}
     \end{figure}

We denote $V_t(\cG)$, $V_p(\cG)$ and $V_d(\cG)$ the numbers of tetrahedral, pillow and double trace vertices of a CTKT graph $\cG$, and $F(\cG)$ the number of faces of $\cG$. The number of pillow vertices splits as the sum of the numbers of pillow vertices of each kind. The edges are not colored, but the faces are colored with a color 1, 2 or 3. Using the same trick as before, we can compute the two point function by reconnecting the external edges and dividing by $N^3$, obtaining $G$ as a sum over rooted graphs:
     \begin{equation}
     \begin{split}
    G = & \sum_{\cG\in \mathfrak{G}} (-\lambda)^{V_t(\cG)}
       \left(-\frac{\lambda_p}{3}\right)^{V_p(\cG)} (-\lambda_d)^{V_d(\cG)} \crcr
    & \qquad \qquad  C^{2 \big[ V_t(\cG)+ V_p(\cG) +V_d(\cG) \big] +1}
       N^{-3 - \frac{3}{2} V_t(\cG) - 2V_p(\cG) - 3V_d (\cG) + F (\cG) } \;.
            \end{split}
     \end{equation}   
As in the case of the colored model, the number of faces of a CTKT graph can be computed in terms of the number of vertices (see Proposition~\ref{prop:CTKTcount} in Appendix \ref{app:deg}):
     \begin{equation}
      F(\cG) =   3 + \frac{3}{2} V_t(\cG) + 2V_p(\cG) + 3V_d(\cG) - 
             \omega(\cG) \; ,
     \end{equation}
 where the  CTKT degree  $ \omega(\cG)$ is a non negative half integer. The two point function (and any other correlation) has a $1/N$ expansion indexed by the CTKT degree:
 \begin{equation}
      G = \sum_{\omega \in \mathbb{N}/2} N^{-\omega} 
      \sum_{\cG\in \mathfrak{G}}^{\omega(\cG) = \omega} (-\lambda)^{V_t(\cG)}
        \left(-\frac{\lambda_p}{3}\right)^{V_p(\cG)} (-\lambda_d)^{V_d(\cG)} 
          C^{2 \big[ V_t(\cG)+ V_p(\cG) +V_d(\cG) \big] +1}
        \;.
 \end{equation}
   
   At leading order one obtains only \emph{melon-tadpole} graphs (see Appendix \ref{app:deg} for details). Similarly to the melonic graphs, the melon tadpole graphs can be seen as a truncation of the self energy depicted in 
   Fig.~\ref{fig:melon-tadpole1}.
  
\begin{figure}[htb]
        \begin{center}
    \psfrag{g}{$G$}
    \psfrag{a}{$\lambda_p,\lambda_d$}
    \psfrag{l}{$\lambda$}
    \psfrag{s}{$\Sigma$}
    \includegraphics[width=0.5\textwidth]{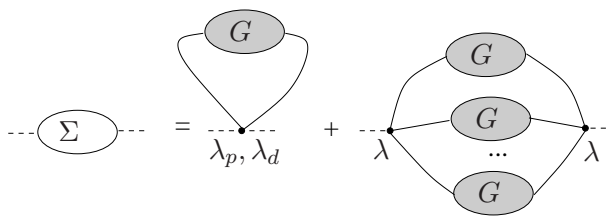}         
    \caption{The melon-tadpole self energy.}\label{fig:melon-tadpole1}       
        \end{center}
\end{figure}

In the melon-tadpole truncation one includes two graphs. The first one is a tadpole graph whose vertex is either a pillow or the double trace such that the edge closes the maximal number of faces. The second is a melon with two tetrahedral vertices. At leading order we have:
\begin{equation}
 (G^{LO})^{-1} = C^{-1} - \Sigma^{LO} \;,\qquad 
 \Sigma^{LO} = - (\lambda_p+ \lambda_d) G^{LO}
  + \lambda^2  (G^{LO})^{3} \;. 
\end{equation}

 \newpage
 
\section{The 2PI formalism}\label{sec:2PI}
\setcounter{equation}{0}

 \bigskip

 It is often convenient to recast a QFT in terms of an effective action \cite{JonaLasinio:1964cw,Cornwall:1974vz}. Effective actions integrate out the quantum fluctuations and the correlation functions are computed by functional derivatives. The most familiar case is the one particle irreducible (1PI) effective action, but the concept is directly generalized to $p$ particle irreducible effective actions \cite{JonaLasinio:1964cw}. 
 
Unsurprisingly, in practice the effective actions are very difficult to compute
and one needs to resort to truncations. The situation is greatly simplified in the case of vector models. The two particle irreducible (2PI) action is particularly well suited for their study because it can be explicitly computed order by order in the $1/N$ expansion \cite{Berges:2004yj}. It turns out that tensor models are similar\cite{Benedetti:2018goh}, and the 2PI action can be computed explicitly at first orders in $1/N$. However, contrary to the vector case, for tensors the effective action is non local at leading order which leads to much richer physics.
  
\subsection{The 2PI effective action}

We introduce some more notation, to be used only in this section. We denote functionals by bold letters and functions by straight letters. We denote
the field $\phi_x$, where $x$ denotes the position and any additional indices. Repeated indices are summed/integrated. Bi local fields are denoted $G_{xy}$,  $J_{xy}$ and so on. A dot denotes integrals  and index contractions. We will sometimes suppress the indices to simplify the notation.

We consider a scalar theory with action and partition function: 
\begin{equation}\label{eq:defmodel2PI}
\pmb{S}[\phi] = \frac{1}{2} \phi \cdot C^{-1} \cdot \phi + \pmb{S}^{\rm int}[\phi] \; , \qquad 
 {\cal Z} =  \int [d\phi] \; e^{ - \pmb{S}[\phi]  } \;.
\end{equation}
The interaction part of the action $ \pmb{S}^{\rm int}[\phi] $ can include
bi valent vertices. They will always be treated as a perturbation of the free theory defined by the covariance $C$. In addition, we require the one point function of the theory to be zero
$
 \Braket{\phi_x} = {\cal Z}^{-1} \int [d\phi] \; e^{ - \pmb{S}[\phi]  } \phi_x  =0 
$.
This is guaranteed if the action is even $\pmb{S}[ - \phi] = \pmb{S}[ \phi]$, which we will assume from now on. Note however that for a colored model the one point function is zero in any rank, even though the action is even only for odd rank. 

In order to define the effective action \cite{Benedetti:2018goh,Benedetti:2019eyl}, we start from the generating function with bi local source term $J_{xy}$:
\begin{equation}
 e^{\pmb{W}[J]}  = \int  [d\phi] \; e^{-\pmb{S}[\phi] + \frac{1}{2} \phi \cdot J \cdot \phi  }   
 \;.
\end{equation}
Observe that, even in the presence of the source, the odd point expectations are zero. The functional $\pmb{W}[J]$ can either be seen as a generating function of connected moments with a bi local source, or as the free energy of the theory with shifted covariance $C^{-1} - J$. It includes the ring graph consisting in only an edge closing onto itself whose amplitude is\footnote{We normalized the integral to $1$ for $\pmb{S}^{\rm int}=0, J=0, C=1$.} 
$-\frac{1}{2} \Tr \ln [  C^{-1} - J ] $. 

The derivatives of $\pmb{W}$ are\footnote{For symmetric functions the derivative is the symmetric projector $\frac{\delta J_{xy}}{ \delta_{J_{ab}}} = \frac{1}{2} \bigg(  \delta_{xa} \delta_{yb} +\delta_{xb} \delta_{ya}  \bigg)\equiv{\cal S}_{xy;ab}$.}:
\begin{align}\label{eq:Wderiv}
  2 \frac{\delta \pmb{W}  }{ \delta J_{xy} } & = \Braket{\phi_x \phi_y}^J = \Braket{\phi_x \phi_y}^J_c \equiv \pmb{G}_{xy} \;,\\
 4  \frac{\delta^2 \pmb{W}}{  \delta J_{xy} \delta J_{ab} } & =
  \Braket{\phi_x\phi_y \phi_a \phi_b}^J - 
  \Braket{\phi_x \phi_y}^J \Braket{\phi_a \phi_b}^J
  = \Braket{\phi_x\phi_y \phi_a \phi_b}^J_c + \pmb{G}_{xa}\pmb{G}_{yb}
   + \pmb{G}_{xb} \pmb{G}_{ya} \;,\nonumber
\end{align}
where the upper script $J$ signifies that that correlations are computed in the presence of the bi local source $J$. The functional $\pmb{G}_{xy}$ (which is a functional of the source $J$) is the connected two point function of the theory with source $J$. Note that the second derivative of $\pmb{W}$ is exactly the four point function $ \Braket{\phi_x\phi_y \phi_a \phi_b}^J_{(xy) \to (ab)} $  we encountered in Section \ref{sec:lect0}. Going  ``on shell''  means putting the source $J=0$. 

We denote $\pmb{J}[G ]$ the  inverse functional of $\pmb{G}[J]$, that is the solution of the equation $\pmb{G}[J]=G$. The effective action is the Legendre transform of $\pmb{W}$ with respect to $J$:
\begin{equation}\label{eq:2PIdef}
\begin{split}
&  \pmb{\Gamma}[G] = -\pmb{W}[ \pmb{J} ] + 
   \frac{1}{2} \Tr[G\pmb{J}]
\;,\crcr 
& \frac{ \delta \pmb{\Gamma} }{\delta G_{xy}}= \frac{1}{2} \pmb{J}_{xy} \; ,\qquad
 \frac{\delta^2 \pmb{\Gamma}}{\delta G_{ab} \delta G_{xy} } = 
 \frac{1}{2} \frac{\delta \pmb{J} }{\delta G } =
 \frac{1}{2} 
 \left( \frac{\delta \pmb{G} }{ \delta J }\right)^{-1}_{J = \pmb{J}} 
  = \frac{1}{4} \left( \frac{\delta^2 \pmb{W} }{\delta J_{xy} \delta J_{ab}}\right)^{-1}_{J = \pmb{J}} \;.
  \end{split}
\end{equation}

This Legendre transform can be written as a functional integral for $\phi$ with inverse covariance $ C^{-1}- \pmb{ J } $ and interaction $\pmb{S}^{ \rm int}[\phi] $:
\begin{equation}
  e^{-\pmb{\Gamma}[G]}   = e^{-\frac{1}{2}\Tr[G\pmb{J}]}\int  [d\phi] \;
  e^{- \frac{1}{2} \phi\cdot (C^{-1} - \pmb{J}) \cdot \phi - \pmb{S}^{\rm int}[\phi] } \;,
\end{equation}
where $\pmb{J}[G]$ is fixed by the condition $\braket{\phi \phi}^{\pmb{J}}_c = G$.

We denote $-\pmb{\Gamma}^{\rm 2PI}[G]$ the generating function of nontrivial 2PI graphs (that is graphs which do not disconnect when cutting two edges) with propagator $G$ and vertices in $\pmb{S}^{ \rm int}[\phi] $. If $ \pmb{S}^{ \rm int}[\phi]$ has bi valent vertices, $ \pmb{\Gamma}^{\rm 2PI}[G]$  contains the graph formed by only one edge with propagator $G$ connected on the bi valent vertex.
For example, in zero dimension with $\pmb{S}^{ \rm int}[\phi] = \frac{m^2}{2} \phi^2 + \frac{\lambda}{4!} \phi^4$ we have:
\begin{equation}
 \pmb{\Gamma}^{\rm 2PI}[G] = \frac{m^2}{2} G + \frac{\lambda}{4!} 3G^2 - \frac{1}{2} 
  \left( \frac{\lambda}{4!} \right)^2 4! G^4 + {\cal O}(\lambda^3)\;.
\end{equation}
The associated graphs are depicted in Fig.~\ref{fig:2PIaction}.
Observe that the mass vertex appears in only one 2PI graph.
\begin{figure}[htb]
        \begin{center}
    \includegraphics[width=0.3\textwidth]{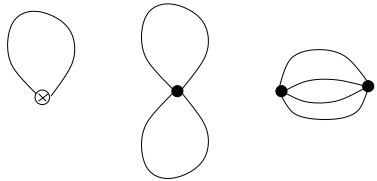}         
    \caption{Graphs contributing to the 2PI action at first orders in a $\phi^4$ theory.}\label{fig:2PIaction}       
        \end{center}
\end{figure}

The derivatives of the 2PI generating function are denoted:
\begin{equation}
 \pmb{\Sigma}[ G]_{xy} = -2 \frac{ \delta  \pmb{\Gamma}^{\rm 2PI} }{\delta G_{xy}}
  \;,\qquad \pmb{ K }[G]_{ab;xy} = G_{aa'} G_{bb'} \frac{\delta \pmb{\Sigma}_{xy}}{ \delta G_{a'b'}} \;.
\end{equation}
$ \pmb{\Sigma}$ is the \emph{self energy} (the amputated one particle irreducible two point function) expressed in terms of the full two point function $G$. To see this we observe that the derivative cuts an edge and $2$ counts the ways to attach it to  the external end points. The fact that this is nothing but the self energy in which all the propagators are fully dressed comes from the remark that the configuration depicted on the left in Fig.~\ref{fig:2Pself} is excluded by the two particle irreducibly condition.
\begin{figure}[htb]
        \begin{center}
    \includegraphics[width=0.6\textwidth]{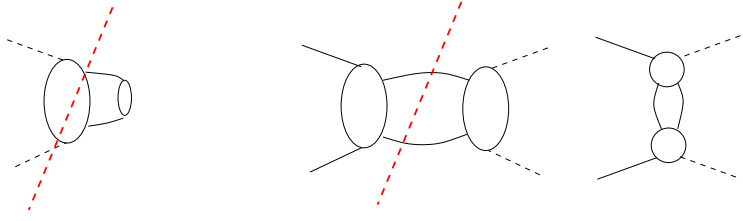}         
    \caption{On the left: a two particle reducible contribution to the self energy. On the right: two contributions to the four point kernel, one which is two particle reducible in the channel $(ab) \to (cd)$ and one which is not (although it is two particle reducible in a different channel. Solid lines represent full two point functions $G$ while dashed lines represent amputations.}\label{fig:2Pself}       
        \end{center}
\end{figure}

The kernel $\pmb{K}$ is the \emph{irreducible four point kernel} in the channel $(ab)\to (xy)$. As it comes from a derivative of $\pmb{\Sigma}$, it can not contain two edges which, when cut, disconnect the kernel into a component having the external points $a,b$ and another component having the external points $x,y$ (this is depicted in Fig.~\ref{fig:2Pself} in the middle). However, the kernel $\pmb{K}$ can be disconnected by cutting two edges in a different channel.

For any source $J$, the two point function is determined self consistently by the Schwinger Dyson equation:
\begin{equation}
  G^{-1} = C^{-1} - J -  \pmb{\Sigma}[ G] \;.
\end{equation}
This equation fixes the source $\pmb{J}[G] = C^{-1} - G^{-1} - \pmb{\Sigma}[ G] $ which ensures that the two point function is exactly $G$. In particular:
\begin{equation}\label{eq:2PIexp}
\frac{ \delta \pmb{\Gamma} }{\delta G} = \frac{1}{2} \pmb{J} 
   =\frac{1}{2} C^{-1} - \frac{1}{2}G^{-1} - \frac{1}{2}\pmb{\Sigma}[G]  \;,
\end{equation}
and, recalling that the self energy is the derivative of the 2PI generating function, this equation can be integrated to obtain\footnote{It is sometimes useful to give a formal functional integral formula for $\pmb{\Gamma}$:
\[
  e^{-\pmb{\Gamma}[G]} 
 = e^{ - \frac{1}{2}\Tr[ C^{-1}  G ] } 
 \int_{2PI} [d\phi] \; e^{- \frac{1}{2} \phi \cdot 
 G^{-1} \phi  - \pmb{S}^{ \rm int}[\phi]} \;.
\]}:
\begin{equation}
\pmb{\Gamma}[G] =  \frac{1}{2}\Tr\big[ C^{-1} G \big]  -   \frac{1}{2} \Tr [\ln ( G ) ] +\pmb{\Gamma}^{\rm 2PI}[G] \;.
\end{equation}

\begin{figure}[htb]
        \begin{center}
    \includegraphics[width=0.9\textwidth]{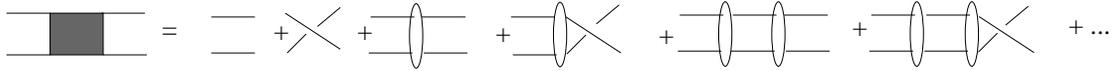}         
    \caption{The four point Dyson equation.}\label{fig:Dyson}       
        \end{center}
\end{figure}

The second derivative of $\pmb{\Gamma}$ is related to the four point kernel:
\begin{equation}
 \frac{\delta^2 \pmb{\Gamma}}{\delta G_{ab}\delta G_{xy}} =  
 \frac{1}{2} G^{-1}_{aa'} G^{-1}_{bb'}  \bigg( {\cal S} - \pmb{ K} \bigg)_{a'b';xy} \;.
\end{equation}
Combining this with Eq. \eqref{eq:2PIdef} and \eqref{eq:Wderiv} we find 
the Dyson equation (see Fig.~\ref{fig:Dyson}) of the four point function connecting $(ab)$ to $(xy)$:
\begin{equation}\label{eq:truc1}
 \Braket{\phi_x\phi_y \phi_a \phi_b}^{\pmb{J}}_{(ab)\to (xy)} =
  \left( \frac{1}{1 - \pmb{K} }\right)_{ab;x'y'} 
  \bigg(   G_{x'x}  G_{y'y} +  G_{x'y}  G_{y'x} \bigg) \;.
\end{equation}

The correlation functions of the original theory are recovered by taking derivatives of $\pmb{\Gamma}$ and then going on shell, that is setting the two point function to be $G^0$, the solution of:
\begin{equation}
\pmb{J}[G^0] = 0 = C^{-1} - (G^0)^{-1}-\pmb{\Sigma}[G^0] \;.
\end{equation}

 \subsection{The Bethe Salpeter equation}

Let $O_{\bar \mu;x}$ be a local operator with some spin:
\begin{equation}
  \left[ \partial_{\mu_1  \dots \mu_{j_1}}(-\partial^2)^{n_1} \phi_x \right] 
    \left[ \partial_{\mu_{j_1 +1}  \dots  \mu_{j_1+j_2}}(-\partial^2)^{n_2} \phi_x \right] \dots 
     \left[ \partial_{\mu_{j_{q-1} + 1}   \dots  \mu_{j_{r-1}+j_r}}(-\partial^2)^{n_r} \phi_x \right] 
    \; .
\end{equation}
We aim to find a closed equation for the three point connected function 
$\Braket{\phi_{x_1} \phi_{x_2} O_{\bar\mu;x}}_c$. To this end we define the generating function with a source $\epsilon^{\bar \mu}_x$ for our  operator:
\begin{equation}
 e^{\pmb{W}^{\epsilon}[J]} = \int [d\phi] \; e^{- \pmb{S}[\phi] + \frac{1}{2}
 \phi \cdot J \cdot \phi + \epsilon^{\bar \mu} \cdot O_{\bar \mu}} \;,
\end{equation}
We assume that the one point expectation $\Braket{\phi_x}$ is zero in the presence of the source. This is always the case if $O$ contains an even number of fields $\phi$. Then: 
\begin{equation}
\begin{split}
 \Braket{O_{\bar \mu; x} }^{J,\epsilon}_c & = \Braket{O_{\bar \mu; x} }^{J,\epsilon} = \frac{\delta \pmb{W}^{\epsilon} }{\delta \epsilon^{\bar \mu}_x} \; , \crcr
  \Braket{\phi_{x_1} \phi_{x_2} O_{\bar \mu ; x}}^{J,\epsilon}_c & = 
    \Braket{\phi_{x_1} \phi_{x_2} O_{\bar \mu ; x} }^{J,\epsilon} -   \Braket{\phi_{x_1} \phi_{x_2} }^{J,\epsilon}_c \Braket{O_{\bar \mu; x}}^{J,\epsilon}_c = 2 \frac{ \delta^2 \pmb{W}^{\epsilon} }{\delta J_{x_1x_2} \delta \epsilon^{\bar \mu}_x} \;. 
\end{split}
\end{equation}
The physical expectations in the theory are obtained by going on shell $J=0$ and setting the source $\epsilon =0$. The results of the previous section go through, except that everything now depends on the source $\epsilon$. We have:
\begin{equation}
 2\frac{\delta \pmb{W}^{\epsilon}}{ \delta J_{x_1x_2} } = \pmb{G}^{\epsilon}[J]_{x_1x_2} \;,\qquad
  \frac{\delta \pmb{G}^{\epsilon}_{x_1x_2}}{\delta \epsilon^{\bar \mu}_{x}} =  \Braket{\phi_{x_1} \phi_{x_2} O_{\bar \mu;x} }_c^{J,\epsilon} \;,
\end{equation}
and we denote $\pmb{J}^{\epsilon}[G]$ the inverse functional of $\pmb{G}^{\epsilon}[J]$
and the Legendre transform of $\pmb{W}^{\epsilon}[J] $ by:
\begin{equation}
\pmb{\Gamma}^{\epsilon}[G] = - \pmb{W}^{\epsilon}[\pmb{J}^{\epsilon}] + \frac{1}{2}\Tr[G \pmb{J}^{\epsilon} ]  \;,\qquad 
 \pmb{\Gamma}^{\epsilon}[G] =  \frac{1}{2}\Tr\big[ C^{-1} G \big]  -   \frac{1}{2} \Tr [\ln ( G ) ] +\pmb{\Gamma}^{\rm 2PI;\epsilon}[G]  \; .
\end{equation}

The functional $- \pmb{\Gamma}^{\rm 2PI;\epsilon}[G] $ is now the sum over 2PI graphs (from the point of view of $\phi$) with propagator $G$ and vertices in $\pmb{S}^{\rm int}[\phi]$ or $\epsilon\cdot O$. With respect to the previous case, we now have additional vertices with coordination $r$ representing insertions of the operator $O$ in the graphs. The derivatives of $\pmb{\Gamma}^{\epsilon}$ can be computed using either the Legendre transform or the explicit formula relating $\pmb{\Gamma}^{\epsilon}$ to $\pmb{\Gamma}^{\rm 2PI;\epsilon}$. In particular:
\begin{equation}
\begin{split}
 - 2 \frac{\delta \pmb{\Gamma}^{\epsilon}}{ \delta G_{x_1x_2} \delta\epsilon^{\bar \mu}_{x}} 
& = -2 \frac{\delta }{ \delta \epsilon^{\bar \mu}_x} 
 \left(  \frac{\delta \pmb{\Gamma}^{\epsilon}}{\delta G_{x_1x_2}} \right)
 = -  \frac{\delta \pmb{J}^{\epsilon}_{x_1x_2}}{ \delta \epsilon^{\bar \mu}_x}  [G]  \;, \crcr  
 - 2 \frac{\delta \pmb{\Gamma}^{\epsilon}}{ \delta G_{x_1x_2} \delta\epsilon_{x}} 
&   =  - 2 \frac{\delta \pmb{\Gamma}^{\rm 2PI; \epsilon}}{ \delta G_{x_1x_2} \delta\epsilon^{\bar \mu}_{x}}  \equiv  
   \Braket{\phi_{x_1} \phi_{x_2} O_{\bar \mu ; x} }^{\pmb{J}^{\epsilon} ,\epsilon }_{2PI}  \; .
\end{split}
\end{equation}
The last correlation $    \Braket{\phi_{x_1} \phi_{x_2} O_{\bar \mu ; x} }^{G;\epsilon}_{2PI} $ is represented in Fig.~\ref{fig:3point2PI}. 
It is two particle irreducible in the channel  $(\phi\phi) \to O$, that is by cutting two edges it can not be disconnected in a component containing both external points $\phi$ and a second connected component containing the operator $O$.
\begin{figure}[htb]
        \begin{center}
    \includegraphics[width=0.4\textwidth]{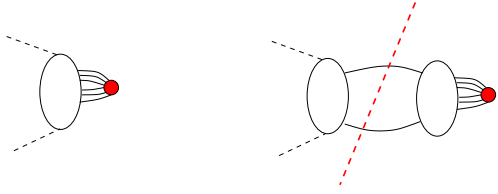}         
    \caption{On the left the correlation function $    \Braket{\phi_{x_1} \phi_{x_2} O_{\bar \mu ; x} }^{G;\epsilon}_{2PI} $. On the right a contribution which is two particle reducible in the channel $(\phi\phi) \to O$. The red dot represents the composite operator.}\label{fig:3point2PI}       
        \end{center}
\end{figure}

On the other, by definition $  \pmb{G}^{\epsilon}[\pmb{J}^{\epsilon}[G]] = G  $
therefore:
\begin{equation}
0=  \frac{\delta \pmb{G}^{\epsilon}_{x_1x_2}}{\delta \epsilon^{\bar \mu}_x }\Big{|}_{ 
 J =  \pmb{J}^{\epsilon} } 
  + \frac{ \delta \pmb{G}^{\epsilon}_{x_1x_2} } { J_{ab} } \Big{|}_{ 
 J =  \pmb{J}^{\epsilon} }   
   \frac{\delta \pmb{J}^{\epsilon}_{ab } } {\delta \epsilon^{\bar \mu}_x}  
   \Rightarrow  \frac{\delta \pmb{J}^{\epsilon}_{ab } } {\delta \epsilon^{\bar \mu}_x} 
    = - 2 \frac{ \delta^2 \pmb{\Gamma}^{\epsilon}}{\delta G_{ab} \delta G_{uv} } \frac{\delta \pmb{G}^{\epsilon}_{uv}}{\delta \epsilon^{\bar \mu}_x }\Big{|}_{ 
 J =  \pmb{J}^{\epsilon} }  \;.
\end{equation}
Putting everything together we conclude that:
\begin{equation}
 \Braket{\phi_{x_1} \phi_{x_2} O_{\bar \mu ; x }}^{\pmb{J}^{\epsilon} ,\epsilon }_{2PI} = 
  2 \frac{ \delta^2 \pmb{\Gamma}^{\epsilon} }{\delta G_{x_1x_2} \delta G_{uv} } 
  \frac{\delta \pmb{G}^{\epsilon}_{uv}}{\delta \epsilon^{\bar \mu}_x }
  \Big{|}_{  J =  \pmb{J}^{\epsilon} }  
 = G^{-1}_{x_1a} G^{-1}_{x_2b} ({\cal S} - \pmb{K}^{\epsilon})_{ab;uv} \Braket{\phi_{u} \phi_{v} O_{\bar \mu ; x} }^{\pmb{J}^{\epsilon} ,\epsilon }_{c} \;,
\end{equation}
which can be rewritten, taking $\epsilon =0$, in the form:
\begin{equation}
 \Braket{\phi _{x_1} \phi_{x_2 } O_{\bar \mu; x} }^{\pmb{J}}_{c} =
    G_{x_1a} G_{x_2b} \Braket{\phi_{a} \phi_{b} O_{\bar \mu;x} }^{\pmb{J}}_{2PI}   +  \pmb{K}_{x_1x_2;ab} \Braket{\phi _{a} \phi_{b } O_{\bar \mu;x} }^{\pmb{J}}_{c}  \; .
\end{equation} 

This equation should be compared with Eq.~\eqref{eq:4pointCFT}:
\begin{equation}
 \int dx^d_3 dx^d_4 \;  K(x_1,x_2 ; x_3, x_4) \Braket{\phi(x_3) \phi(x_4) 
    O_{\Delta,J}^{\bar \mu}(x) } = k(\Delta,J) 
 \Braket{\phi(x_1) \phi(x_2) 
    O_{\Delta,J}^{\bar \mu}(x) } \; ,
\end{equation}
together with the condition that the physical primary operators are such that $k(\Delta,J)=1$. This means that in a CFT the 2PI contribution to the three point function of two $\phi$ fields and a primary $O_{\bar \mu}$ must be identically zero. The implications of this fact need to be investigated in depth.

\subsection{The $1/N$ expansion and melonic theories}

We consider the generalization of the $O(N)^3$ model to dimension $d$. 
Using the notation of Section~\ref{sec:lect1}. The field is a tensor with three indices $\phi_A(x)$ and the action of the $O(N)^3$ field theory is:
     \begin{equation}
     \begin{split}
      \pmb{S}[\phi] =&  \frac{1}{2} \int_{xy} \phi_{A}(x) ( C^{-1})(x,y)\phi_A(y)  +
         \frac{1}{2}   m \int_x   \phi_{A}(x) \phi_{A}(x)  \crcr
        & +
       \int_x \left(  \frac{\lambda}{4N^{3/2}} 
      \delta^t_{ABCD} + \frac{\lambda_p}{4N^2} \delta^p_{AB;CD} +
        \frac{\lambda_d}{4N^3} \right)
         \phi_{A}(x) \phi_{B}(x) \phi_{C}(x) 
         \phi_{ D }(x)  \; ,
    \end{split}
     \end{equation}
 where from now on we reinstate the separation between the position arguments and the tensor indices. We added a mass parameter $m$ and the covariance of the theory, $C$, is kept arbitrary for now.  
     
 The source $J_{AB}(x,y)$ is bi local both in positions and tensor indices. In order to determine the scaling in $N$ of a term in the 2PI action, we use the diagonal ansatz $G_{AB}(x,y) \sim \delta_{AB}G(x,y)$ because on shell the two point function is indeed diagonal in the tensor indices. It follows that the scaling in $N$ of the 2PI graphs is just the standard scaling in $N$ discussed in Section~\ref{sec:lect1}. At leading order only melon--tadpole graphs contribute, and the only 2PI melon--tadpole graphs are those represented in Fig.~\ref{fig:2PIaction}. This is because an insertion of a melon or a tadpole in any of the three graphs yields a two particle reducible contribution. At leading order in $1/N$ we get:
\begin{equation}
\begin{split}
   \pmb{\Gamma}^{2PI}[G] = &   \frac{m }{2} \Tr[G]  + \int_x G_{AB}(x,x)  \left(  \frac{\lambda_p}{4N^2} \delta^p_{AB; CD}  +  \frac{\lambda_d}{4N^3}  \delta^d_{AB;CD} \right) G_{CD}(x,x)   \crcr
  &\qquad - \frac{1}{2} \left( \frac{\lambda}{4 N^{3/2}}\right)^2 4 \int_{x,y} 
  \delta^t_{A_1A_2A_3A_4 } \delta^t_{B_1B_2B_3B_4 } \prod_{i}G_{A_iB_i}(x,y) 
  \; .
\end{split}
\end{equation}
The self energy is then:
\begin{equation}
\begin{split}
  \pmb{\Sigma}[G]_{AB}(x,y) = & -m  \delta_{AB} \delta_{xy} - 
  \left( \frac{\lambda_p}{N^2}   \delta^p_{AB ; CD} + \frac{\lambda_d}{N^3}\delta^d_{AB ; CD}   \right) \delta_{xy} G_{CD}(x,x)   \crcr
  &\qquad + \frac{\lambda^2}{N^3} \delta^t_{AA_1A_2A_3 }  \delta^t_{B B_1B_2B_3} 
    \prod_{i=1}^3 G_{A_iB_i}(x,y)  \;,
\end{split} 
\end{equation}
and the irreducible kernel is:
\begin{align} \label{eq:kernel}
& \pmb{ K }_{A'B';CD } (x'y'; zt) = \\
& =  G_{A'A}(x',x) G_{B'B}(y',y)  \bigg[
    - \left( \frac{\lambda_p}{N^2} \delta^p_{AB;CD} + \frac{\lambda_d}{N^3}  \delta^d_{AB;CD} \right) \delta_{xy} \delta_{xz} \delta_{xt}  \crcr
 & \qquad + \frac{\lambda^2}{N^3} \delta^t_{AA_1A_2A_3 } 
 \delta^t_{BB_1B_2B_3 }    \sum_{i=1}^3
     \left( \frac{1}{2} \delta_{xz}\delta_{yt} \delta_{A_i C} \delta_{ B_i D} + \frac{1}{2}\delta_{xt}\delta_{yz}  \delta_{A_i D} \delta_{B_i C}  \right) 
    \prod_{j\neq i} 
    G_{C_jD_j}(x,y) \;. \nonumber
\end{align}

  \newpage

\section{Renormalization in a tensor field theory}\label{sec:ren}
\setcounter{equation}{0}

 \bigskip
 
In Sections~\ref{sec:lect0} and \ref{sec:lect1} we have seen that melonic CFTs can be analytically treated and that in many models the melonic limit can be recovered as a large $N$ limit. CFTs should correspond to fixed points of the renormalization group and infrared fixed points are especially interesting because they describe the low energy behavior of theories.

The natural question is: are there examples of field theories with infrared attractive fixed points described by melonic CFTs? The answer to this question is yes: many fermionic models in less than 2 dimensions \cite{Witten:2016iux,Prakash:2017hwq,Benedetti:2017fmp,Klebanov:2018nfp,Pakrouski:2018jcc,Klebanov:2019jup}, and some supersymmetric \cite{Popov:2019nja}
or bosonic ones \cite{Giombi:2018qgp} in dimension strictly less that 3 do have melonic infrared fixed points.

However, it turns out that it is not so easy to find models with melonic fixed points in $d=3$ dimensions. In this section we discuss one example which works \cite{Benedetti:2019eyl}. 

From now on we consider $d<4$ dimensions. Although we keep $d$ generic, we are especially interested in the $d=3$ case.
Our starting point is the $O(N)^3$ field theory described in section \ref{sec:2PI},
with a suitable covariance:
      \begin{equation}\label{eq:actionfin}
     \begin{split}
      \pmb{S}[\phi] =&  \frac{1}{2} \int d^d x  \;  \phi_{A}(x)
      ( -\partial^2)^{\zeta} \phi_A(x)  +
         \frac{1}{2}   m \int d^dx\;   \phi_{A}(x) \phi_{A}(x)  \crcr
        & +
       \int d^dx \; \left(  \frac{\lambda}{4N^{3/2}} 
      \delta^t_{ABCD} + \frac{\lambda_p}{4N^2} \delta^p_{AB;CD} +
        \frac{\lambda_d}{4N^3} \right)
         \phi_{A}(x) \phi_{B}(x) \phi_{C}(x) 
         \phi_{ D }(x)  \; ,
    \end{split}
     \end{equation}
where for now $\zeta$ is not fixed and  we take the large $N$ limit.

The first choice that comes to mind  \cite{Giombi:2017dtl} is 
$\zeta=1$, that is the tensor generalization of the 
standard $\phi^4_4$ theory. The quartic couplings are classically marginal in dimension $d=4$, and one can search for Wilson Fisher \cite{Wilson:1971dc} like fixed point in $d=4-\epsilon$ dimensions \cite{Giombi:2017dtl}. At first orders one finds the beta functions
(where $\tilde g = g/(4\pi)^2$):
\begin{equation}\label{eq:WFKlebanov}
\begin{split}
\beta_{\tilde g} &= -\epsilon \tilde g + 2\tilde g^3 \;,\qquad
\beta_{\tilde g_p} = -\epsilon \tilde g_p + 
\left( 6\tilde g^2 +\frac{2}{3} \tilde g_p^2\right) - 2  \tilde g^2\tilde g_p \;, \crcr
\beta_{\tilde g_d} &= -\epsilon \tilde g_d + 
\left( \frac{4}{3}\tilde g_p^2 + 4 \tilde g_p\tilde g_d + 2\tilde g_d^2 \right)
  - 2\tilde g^2 ( 4\tilde g_p + 5\tilde g_d ) \;,
\end{split}
\end{equation}
which admit a fixed point 
$g_{\star} = ( \epsilon /2 )^{1/2}, \; 
  g_{p,\star} = \pm 3\im  ( \epsilon /2 )^{1/2} ,\; 
 g_{d,\star} = \mp \im (3\pm \sqrt{3} ) ( \epsilon /2 )^{1/2}$. 

 The pillow and double trace couplings are purely imaginary at the fixed point. This in itself is not a problem, but a more careful study reveals other unpleasant features of the fixed point: searching for the dimensions of the physical primary fields at this fixed point along the lines of Section \ref{sec:lect0}, one finds a primary with complex dimension $d/2 \pm \im \alpha$ \cite{Giombi:2018qgp}. Now, this is problematic:
\begin{itemize}
  \item in an  AdS/CFT  picture \cite{Gubser:1998bc,Witten:1998qj}, bulk fields with mass $m_{AdS}$ corresponds to boundary single trace primaries with dimensions $\Delta=  d/2 \pm (d^2 / 4 + m_{AdS}^2)^{1/2}  $. Primaries with dimensions $d/2 \pm \im \alpha$ correspond to fields with 
  $m_{AdS}^2< -d^2/4$ violating the Breitenlohner--Freedman \cite{Breitenlohner:1982jf} bound.
  \item a physical primary with dimension $d/2\pm\im \alpha$ represents a pole of the density $\rho(\Delta,J)$ located exactly on the original contour of integration of the partial waves (recall Section \ref{sec:lect0}). The initial expansion of the four point function in terms of partial waves needs to be revisited in order to deal with this singularity.
  \item the problematic primary is the mass operator. A dimension of a mass--like primary operator of the form $d/2\pm\im \alpha$ has recently been shown in a similar model \cite{Kim:2019upg} to correspond to an instability and signal that the corresponding operator acquires a non zero vacuum expectation value. 
 \item the dimension of the mass is half the one of the double trace invariant \cite{Gubser:2002vv} which is $d+ \nu$, with $\nu$ the critical exponent of the double trace coupling. The dimension $d/2 \pm \im \alpha$ of the mass implies that the double trace coupling has a purely imaginary critical exponent. The fixed point is a limit cycle, not an infrared fixed point (see Appendix\ref{app:ren}).
 \end{itemize}

 In $4+\epsilon$ dimensions the problem goes away, but the fixed point turns out to be an ultraviolet fixed point: the pillow and double trace couplings are relevant at the fixed point. 
 
 In order to find a genuine infrared fixed point described by a melonic CFT one needs to take a more drastic approach.
 According to Eq.~\eqref{eq:meloSDE}, in the melonic limit the field is expected to acquire the infrared scaling dimension $\Delta_{\phi} = d/q$, with $q=4$ in our case. The idea \cite{Benedetti:2019eyl} is to modify the ultraviolet scaling of the free part of the action $\zeta$ in such a way that the UV scaling dimension of the field  
 $( d-2\zeta ) /2 $ equals the IR one. 
 
 From now on we fix $\zeta = d /2 -d/q$, which is $\zeta = d/4$ for $q=4$. 
 This means that the free part of the action has a non integer power of the momentum. 
  Before continuing, let us briefly comment on this. Although models with non integer scaling have been considered in the literature \cite{Brydges:2002wq,Abdesselam:2006qg} (and more recently in \cite{Gross:2017vhb} in the context of the SYK model), they might be somewhat unfamiliar to the reader. 
  
  For any $\zeta \le 1$ the free theory:
  \begin{equation}\label{eq:libzeta}
   \pmb{S}_0[\phi] = \frac{1}{2} \int d^dx \; \phi(x) (-\partial^2)^{\zeta} \phi(x) \; , \qquad \zeta\le 1 \;,
 \end{equation}
 is \emph{unitary} because it is explicitly Osterwalder Schrader positive. Indeed, the covariance:
\begin{align}
C(p)  & = \frac{1}{p^{2\zeta}} = \frac{1}{\Gamma(\zeta)} \int_0^{\infty} 
d\alpha \; \alpha^{\zeta-1} e^{-\alpha p^2} \;, \\
C(x-y)  &  =  \frac{1}{ (4\pi)^{d/2} \Gamma(\zeta)} \int_0^{\infty}  d\alpha \; \alpha^{\zeta-1-d/2} e^{- \frac{(x-y)^2}{4\alpha}}   =  \frac{\Gamma\left( \frac{d}{2} -\zeta \right) }{2^{2\zeta} \pi^{d/2} \Gamma(\zeta) }\;\frac{1}{|x-y|^{d-2\zeta}}\; ,\nonumber
\end{align}
 (where the last equality holds for $d-2\zeta>0$) admits an absolutely convergent K\"all\'en--Lehmann spectral representation as a superposition of massive particles with a continuous mass spectrum:
\begin{equation}
  \frac{1}{p^{2\zeta}} = \frac{1}{\Gamma(\zeta) \Gamma(1-\zeta)}\int_0^{\infty}dx \; \frac{ x^{-\zeta} }{p^2 + x} \; .
\end{equation}
The condition $\zeta<1$ is crucial for the convergence in $0$. 

One can also consider interacting theories with $\zeta<1$  \cite{Brydges:2002wq,Abdesselam:2006qg}. The most well known example is the Brydges--Mitter--Scoppola model with $d=3$,  $\zeta = 3/4+\epsilon$ and $\lambda\phi^4$ interaction. This model has:
\begin{itemize}
 \item the Gaussian fixed point where the quartic coupling is relevant,
 \item an interacting fixed point $g_{\star} \sim \epsilon$ (with $g$ the running dimensionless quartic coupling) where the quartic coupling is irrelevant,
 \item a renormalization group trajectory connecting the two fixed points.
\end{itemize}
These statements can are rigorously proven  \cite{Brydges:2002wq,Abdesselam:2006qg}. 
The infrared fixed point of this model is the inspiration for using a non integer scaling in our case.

\subsection{Renormalization}

Although $q$ is fixed to 4, we will often keep it generic. There are two reasons for this. First, this makes the comparison with Section \ref{sec:lect0} easier. Second, it is likely that some of the results can be generalized for colored models with $q$ body interactions. 
As $\zeta = d/2-d/q$ and $\zeta<1$ we obtain a bound $d < 2q/(q-2)$ that is:
\begin{itemize}
 \item[-] for $q=4$, $d < 4$. In particular this covers a quartic model in $d=3$, our main interest. The case $4-\epsilon$ can also be recovered.
 \item[-] for $q=6$, $d<3$.
 \item[-] any $q$ in $d=2$.
\end{itemize}

Following Appendix \ref{app:ren}, we introduce an ultraviolet cutoff $\Lambda$ and an infrared cutoff $k$:
\be
 C^{\Lambda}_k = \frac{1}{p^{2\zeta}} \chi^{\Lambda}_k(p) = 
    \frac{1}{\Gamma(\zeta)} \int_{\Lambda^{-2}}^{k^{-2}} d\alpha \;
     \alpha^{\zeta -1} \;e^{-\alpha p^2} \;, 
\ee
that is we chose a multiplicative cutoff $\Theta( u ) =\Gamma(\zeta)^{-1} \int_u^{\infty} d\alpha \; \alpha^{\zeta-1} e^{-\alpha}$ which is
an upper incomplete Euler Gamma function. We aim to compute the wave function renormalization (and consequently the anomalous field dimension) and the 
$\beta$ functions of the couplings.

We start from the large $N$ self energy and four point kernel. Using Section \ref{sec:2PI}, on shell we have $G_{AB}(x,y) = G_{xy} \delta_{AB}$ and:
 \begin{equation}
 \begin{split}
  &  \Sigma_{AB}(x,y)  = \delta_{AB} \Sigma_{xy}
   \;,\qquad \Sigma_{xy}= -m \delta_{xy} - (\lambda_p + \lambda_d) \delta_{xy}
     G_{xx}+ \lambda^2 G_{xy}^3  \;, \crcr
&  K_{AB;CD}(x'y';zt)  = G_{x'x} G_{y'y} \crcr
  &
  \qquad \qquad \bigg[  - \frac{\lambda_p}{N^2} \delta^p_{AB;CD}\delta_{xy} \delta_{xz} \delta_{zt}    - \frac{\lambda_d}{N^3} \delta^d_{AB;CD}\delta_{xy} \delta_{xz} \delta_{zt}  
     + \frac{\lambda^2}{N^2} 3 \delta^p_{AB;CD} \delta_{xz} \delta_{yt} G_{xy}^2 \bigg] \; ,
      \end{split}
 \end{equation}
where $m$ and $\lambda$ denote the dimensionful mass parameter and four point couplings at the UV scale $\Lambda$. 
 It is convenient to parametrize the interaction in terms of $\lambda_1 = \lambda_p/3$ and $\lambda_2 = \lambda_p+\lambda_d$ and the two mutually orthogonal orthogonal projectors $P_1 = 3 \left( \frac{1}{N^2} \delta^p - \frac{1}{N^3} \delta^d\right)$ and $P_2= \frac{1}{N^3} \delta^d$. In momentum space we get:
 \begin{align}\label{eq:onsheellker}
  & \Sigma(p)   = -m  - \lambda_2 \int_r G(r) + \lambda^2
     \int_{r_1r_2}G(r_1) G(r_2) G(p+r_1+r_2) \\
& K_{ p_1,p_2; r_1,r_2 }  =  (2\pi)^d \delta(p_1+p_2 - r_1 -r_2)   G(p_1) G(p_2) \crcr  
& \qquad \qquad \bigg[ \bigg( \lambda^2 \int_r G(r)G(r+p_1-r_1)       -\lambda_1 \bigg) P_1 + \bigg( 3\lambda^2 \int_r G(r)G(r+p_1-r_1) -  \lambda_2 \bigg) P_2 
\bigg]\; . \nonumber
 \end{align}
Note that $\lambda_1$ is essentially the pillow coupling and $\lambda_2$ essentially the double trace one.

 \paragraph{The wave function.} In momentum space the Schwinger Dyson equation with cutoffs becomes:
\begin{equation}\label{eq:SDEcutoff}
G_k(p)^{-1} =p^{2\zeta} \chi^{-1} +  m + \lambda_2 \int_r G_k(r) - \lambda^2
     \int_{r_1\dots r_{q-2}}G_k(r_1) \dots G_k(r_{q-2}) G_k(p+r_1+ \dots +r_q) \;,
\end{equation}
where in the last term we reintroduced a generic $q$. It turns out that
(after tuning the bare mass):
\begin{equation}
  G_k(p) = \frac{1}{Z p^{2\zeta}} \, \chi^{\Lambda}_k(p) \;,  \qquad
  \zeta = \frac{d}{2} -\frac{d}{q} \;,
\end{equation}
with a constant $Z$ (to be determined below) verifies Eq.~\eqref{eq:SDEcutoff} up to terms which vanish in the limit $k\to 0$. The important point is that the total $Z$ is finite, hence the anomalous field dimension $\eta_{\star}$ is zero. This is consistent with the the fact that $\zeta$ has been chosen such that the ultraviolet and infrared scaling dimensions of the field coincide.

To check this, let us first consider the local part of the right hand side of Eq.~\eqref{eq:SDEcutoff} (we need to remember that $q=4$ hence $\zeta = d/4$ for this discussion):
\begin{equation}
 m + \lambda_2 \int_r G_k(r) - \lambda^2
     \int_{r_1 r_2 }G_k(r_1)  G_k(r_2) G_k(r_1+ r_2) \;.
\end{equation}
The first term is $
 \int_r G_k(r) \sim k^{d/2} -\Lambda^{d/2}$, which vanishes in the $k\to 0$ limit if 
 $m=-\lambda_2 \Lambda^{d/2}$. The second term is similar. 
 
 Once the local part of the SDE is subtracted via a Taylor expansion with integral rest \cite{Benedetti:2019eyl} we can take the cutoffs to their limits and, rescaling the $\alpha$ by $p^2$, we obtain:
\begin{equation}
\begin{split}
  Z p^{2\zeta} = & p^{2\zeta} + p^{2\zeta } \frac{\lambda^2}{ (4\pi)^{d\frac{q-2}{2}} \Gamma(\zeta)^{q-1} Z^{q-1}} \crcr
  & \qquad \qquad \qquad\int_0^1 dt 
  \int_{0}^{\infty} d\alpha \; 
  \frac{ \prod_{i=1}^{q-1}  \alpha_i^{\zeta} }{
  \left( \sum_{i=1}^{q-1} \prod_{j\neq i} \alpha_j  \right)^{d/2+1}} 
 e^{ -t \frac{\prod_{i=1}^{q-1}  \alpha_i }{ \sum_{i=1}^{q-1} \prod_{j\neq i} \alpha_j  }  }  \;.
\end{split}
\end{equation}
Using Appendix \ref{app:wf} we see that the total wave function renormalization $Z$ verifies the equation:
\begin{equation}\label{eq:Z}
   1 = \frac{1}{Z} + \frac{1 }{ g_c^2 }  
  \left( \frac{ \tilde \lambda}{Z^{q/2}}\right)^2 \; , \quad
  \tilde \lambda \equiv \frac{\lambda}{ (4\pi)^{\frac{d(q-2)}{4}  } \Gamma(\zeta)^{q/2}}
   \;,\quad \frac{1}{g_c^2} = 
   \frac{\Gamma(\zeta) \Gamma(1-\zeta) \Gamma\left( \frac{d}{2} -\zeta \right)^{q-1}  }{\zeta \; \Gamma\left( \frac{d}{2} + \zeta\right) } \;. 
\end{equation}

It is instructive to compute the two point function in direct space
$G(x_{12}) = b |x_{12}|^{- 2\Delta_{\phi} } $. Taking the Fourier transform and recalling that $\Delta_{\phi} = d/q$ we obtain that $b$ verifies:
 \begin{equation}\begin{split}
    1 = b \frac{2^{d-2\Delta_{\phi}}\pi^{d/2}\Gamma(\frac{d}{2} -\Delta_{\phi})}{\Gamma(\Delta_{\phi} )}
    + \lambda^2b^q  \pi^d \frac{\Gamma\left( 1 - \frac{d}{2} + \Delta_{\phi}\right) \Gamma\left(\frac{d}{2} -\Delta_{\phi}\right)}
    {\left(\frac{d}{2} -\Delta_{\phi} \right) \Gamma(d-\Delta_{\phi}) \Gamma(\Delta_{\phi})} \;,
                 \end{split}
  \end{equation}
which reproduces Eq.~\eqref{eq:meloSDE} if one neglects the first term on the right hand side.
 
\paragraph{Four point couplings.} 
From now on we  denote:
\begin{equation}
 \tilde \lambda \equiv \frac{\lambda}{ (4\pi)^{\frac{d(q-2)}{4}  } \Gamma(\zeta)^{q/2}} \;, \qquad 
  \tilde \lambda_{i} \equiv \frac{\lambda_i}{ (4\pi)^{d/2} \Gamma(\zeta)^{2}} 
  \; .
\end{equation}

The classical scaling dimension of an operator $\partial^J \phi^n$ is 
$J + n\Delta_{\phi}$ (see Appendix \ref{app:ren}). In our case $\Delta_{\phi}=d/q$ and the only classically marginal operators are $J=0, n=q$. For $q=4$ they are the tetrahedron, pillow and double trace. 

At leading order in $1/N$ the tetrahedral coupling does not receive any radiative correlation  therefore the renormalized tetrahedral coupling is just a rescaling of the bare one by the wave function constant. Using the tilde  couplings we write:
\begin{equation}
  \tilde g = \frac{\tilde \lambda}{Z^{q/2}}   \;.
\end{equation}

The renormalized tetrahedral coupling does not run, hence it is just a parameter which can be adjusted. On the contrary, the pillow and double trace couplings do run. We denote $\tilde g_{1}, \tilde g_{2}$ the running dimensionless couplings at scale $k$ (we suppress the dependence in $k$ in order to simplify the notation). We will still keep $q$ generic in some formulae, but we will remember that the pillow and double trace couplings are four point couplings. The $\tilde g_i$s are the amputated 1PI four point functions at zero momentum divided by $Z^2$. In terms of the four point kernel Eq.~\eqref{eq:onsheellker} we get:
 \begin{equation}\label{eq:startbare}
 \tilde g_{i} = \frac{ \Gamma^{4;i} }{  (4\pi)^{ d/2 } \Gamma(\zeta)^{2}   Z^{2} } \;,\qquad 
  -\Gamma^{4;1} P_1  -\Gamma^{4;2} P_2 = G^{-1}G^{-1} \frac{K}{1-K} \;.
  \end{equation}
As $P_1$ and $P_2$ are mutually orthogonal, the two cases $i=1,2$ are identical up to substituting $\tilde \lambda^2 $ by $(q-1) \tilde \lambda^2$. 

Expanding the series in 
Eq.~\eqref{eq:startbare} we obtain the bare expansion: $\tilde g_1$ is a sum over ``sausage graphs'' depicted in Fig.~\ref{fig:bare1}. A sausage graph  is a sequence of vertical irreducible pieces connected by pairs of horizontal edges. The vertical pieces are either ladder rungs with two tetrahedral couplings or bare vertices $\lambda_1$. 

\begin{figure}[ht]
\begin{center}
\includegraphics[scale=1.3]{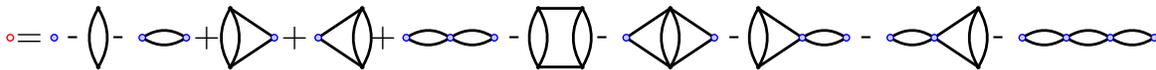} 
 \caption{The bare series up to quartic order (the blue vertices represent $\lambda_1$).} \label{fig:bare1}
 \end{center}
\end{figure}

Each graph has a (log divergent) amplitude:
\begin{equation}\label{eq:ampli1}
 A(\cG) = \int_{\Lambda^{-2}}^{k^{-2}} 
 \left( \prod_{e\in \cG} d\alpha_e \;  \alpha_e^{\zeta-1}\right) \; 
 \frac{ 1 }{ \big[ \sum_{ {\cal T} \subset \cG} \prod_{e\notin {\cal T} } \alpha_e \big]^{d/2} } \;,
\end{equation}
where $e\in \cG$ runs over the edges of $\cG$ and ${\cal T}$ over the trees in $\cG$ \cite{Rivasseau:1991ub,Krajewski:2008fa}. The graph consisting in only a bare vertex has amplitude $1$. We denote $\mathfrak{S}$ the set of connected sausage graphs with at least two internal vertices, and $n_t(\cG)$ respectively $n_1(g)$ the numbers of tetrahedral vertices and blue vertices of $\cG$. Then:
\begin{align}
  \tilde g_1( \tilde \lambda_1, \tilde \lambda)  =   \frac{\tilde \lambda_1}{Z^2}  
  +  \sum_{\cG\in \mathfrak{S}} (-1)^{1 + n_1(\cG) } 
     \left(  \frac{ \tilde \lambda}{Z^{q/2}}\right)^{n_{t}(\cG)} 
 \bigg( \frac{ \tilde \lambda_1}{Z^2}  \bigg)^{n_1(\cG)} 
 A(\cG)\; .
\end{align}
Observe that this is naturally a series in the renormalized tetrahedral coupling $\tilde g = Z^{-q/2} \tilde \lambda$. The bare expansion for $g_2$ is identical up to replacing $\tilde \lambda^2$ by $(q-1) \tilde \lambda^2$. 

The graphs with no internal $\lambda_1$ vertex are special. They might have no external $\lambda_1$ vertex either (ladders), or one external $\lambda_1$ vertex (caps) or two external $\lambda_1$ vertices (double caps), as depicted in Fig.~\ref{fig:1VR}.

\begin{figure}[htb]
\begin{center}
\includegraphics[scale=0.7]{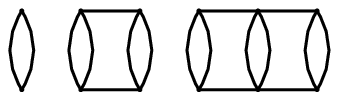} 
\hspace{40pt}
\includegraphics[scale=0.7]{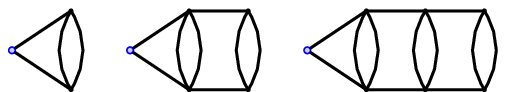} 
\hspace{40pt}
\includegraphics[scale=0.7]{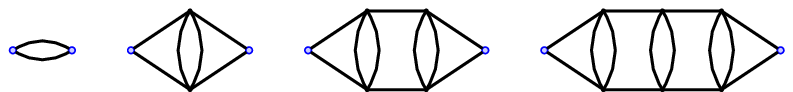} 
 \caption{Ladders, caps and double caps.} \label{fig:1VR}
 \end{center}
\end{figure}

Let us denote $U_r,S_r,T_r$ the amplitude of the ladder, cap respectively double cap with $2r$ tetrahedral vertices, and let us define the generating functions:
\begin{equation}
    U(\tilde g) =  \sum_{r\ge 1} \tilde g^{2r} U_r  \;,\qquad
     S(\tilde g) =  \sum_{r \ge 1 } \tilde g^{2r} S_r  \;,\qquad   
     T (\tilde g) =  \sum_{r\ge 0} \tilde g^{2r } T_r \; ,
\end{equation}
where, in $S(\tilde g)$ and $T(\tilde g)$, we have \emph{not} included any coupling constants for the $\lambda_1$ vertices. The crucial fact is that the amplitude of any sausage factors at the internal $\lambda_1$ vertices, thus:
\begin{align}\label{eq:bare1VR}
 \tilde g_1 = &- U( \tilde g) + \bigg( \frac{ \tilde \lambda_1}{Z^2} \bigg) 
\frac{\big[ 1+S( \tilde g) \big]^2}
 { 1 + \frac{ \tilde \lambda_1 }{Z^2} T( \tilde g ) } \;. 
\end{align}
This gives a particularly simple $\beta$ function and short computation yields:
\begin{equation}\label{eq:beta}
 \beta_{\tilde g_1} = k\partial_k \tilde g_1
 = \beta_0^{\tilde g}
  - 2  \beta_1^{\tilde g} \; \tilde g_1  + \beta_2^{\tilde g} \; \tilde g_1^2\;,
\end{equation}
with the coefficients of the $\beta$ function given by:
\begin{equation}\label{eq:coef}
 \begin{split}
 \beta_0^{\tilde g}  &= - k\partial_k  U + 2 \frac{U}{1+S}k\partial_k  S - \frac{U^2}{(1+S)^2} k\partial_k  T   \; , \qquad  \beta_2^{\tilde g} = - \frac{1}{(1+S)^2}  k\partial_k T
\;, \\
  \beta_1^{\tilde g} & = - \frac{1}{1+S}  k\partial_k S + \frac{U}{(1+S)^2}  k\partial_k  T \; .  
 \end{split}
\end{equation}

This result is an all order result in the couplings: this is the complete $\beta$ function at leading order in $1/N$. The important remark is that the $\beta$ functions are quadratic.  The coefficients $ \beta_{0,1,2}^{\tilde g}$ are series in the tetrahedral coupling $\tilde g$. While it is not obvious, they are finite term by term in the limit $k\to 0$ \cite{Benedetti:2019eyl}.   
 
\subsection{Fixed points}

Let us recapitulate where we stand. Starting with the UV scaling $\zeta = d/2-d/q<1$ hence field dimension $\Delta_{\phi}=d/q$ we obtained the following results.

\paragraph{\emph{Wave function.}} Tuning the renormalized mass to zero and lifting the cutoffs the two-point function is: 
 \begin{equation}
  G(p)  = 
  \frac{1}{Z p^{2\zeta}} \;, \qquad   1 = \frac{1}{Z} + \frac{\tilde g^2}{g_c^2} 
   \;, \qquad 
      \frac{1}{g_c^2} = 
   \frac{\Gamma(\zeta) \Gamma(1-\zeta) \Gamma\left( \frac{d}{2} -\zeta \right)^{q-1}  }{\zeta \; \Gamma\left( \frac{d}{2} + \zeta\right) } 
   \;,
 \end{equation}
that is the anomalous field dimension $\eta_{\star}$ is $0$. 

\paragraph{\emph{Tetrahedral coupling.}} The tetrahedral coupling has a finite flow, that is the renormalized coupling is a rescaling of the bare one $ \tilde g = Z^{-q/2} \tilde \lambda$. In particular we have:
  \begin{equation}
   Z= \left( 1 - \frac{ \tilde g^2}{g_c^2} \right)^{-1} \;, \qquad \lambda = \tilde g Z^2   
   \;.
  \end{equation}
   There are two cases, depicted in Fig.~\ref{fig:tetra}: $\tilde \lambda$ real and $\tilde \lambda$ purely imaginary:
   \begin{itemize}
    \item {\emph{$\tilde \lambda$ (and $\tilde g$) real:}} $\tilde\lambda(\tilde g)$ is invertible to $\tilde g(\tilde \lambda)$ for any $\tilde\lambda $ and $g< g_c$. 
    \item {\emph{$\tilde\lambda$ (and $\tilde g$) imaginary:}} $\tilde \lambda(\tilde g)$ is invertible to $\tilde g(\tilde\lambda)$ 
    for $|\lambda |< 3^{3/2}2^{-4}g_c$ and  $| g | < 3^{-1/2}g_c$ (the end point of the blue curve in  Fig.~\ref{fig:tetra}.
   \end{itemize}
\begin{figure}[htb]
\begin{center}
\includegraphics[scale=0.2]{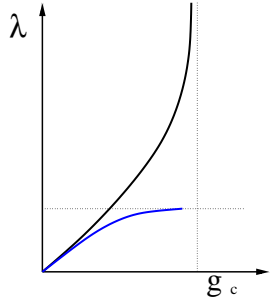} 
 \caption{Bare and renormalized tetrahedral couplings (in blue
 the purely imaginary case).}\label{fig:tetra}
 \end{center}
\end{figure}
  
\paragraph{\emph{Pillow and double-trace couplings.}} At leading order in $1/N$ but at all orders in the couplings the $\beta$ functions are:
  \begin{align}
 \beta_{\tilde g_1} & =  k \partial_k \tilde g_1 
   =  \beta_0^{\tilde g} - 2 \beta_1^{ \tilde g} \, \tilde g_1 +
   \beta_2^{\tilde g} \, \tilde g_1^2 \;, \qquad
    \crcr
  \beta_{ \tilde g_2} & =  k\partial_k \tilde g_2
   =  \beta_0^{ \sqrt{(q-1)} \tilde g }- 2 \beta_1^{\sqrt{(q-1)} \tilde g } \,  \tilde g_2 + \beta_2^{ \sqrt{q-1} \tilde g }\,  \tilde g_2^2 \;,
  \end{align}
where $   \beta_0^{ \tilde g}, \beta_1^{ \tilde g } $  and   $\beta_2^{ \tilde g}$ are power series in $\tilde g^2$. At first orders they are \cite{Benedetti:2019eyl}:
\begin{equation}
 \beta_0^{\tilde g} = 
 \left( 2\frac{\Gamma(\frac{d}{4})^2}{ \Gamma(\frac{d}{2})}\right) \tilde g^2+  \mathcal{O}(\tilde g^4) \;,\qquad
 \beta_1^{ \tilde g} =  \mathcal{O}(\tilde g^2) \;, \qquad
  \beta_2^{ \tilde g} = \left( 2\frac{\Gamma(\frac{d}{4})^2}{ \Gamma(\frac{d}{2})}\right) + \mathcal{O}(\tilde g^2) \;.  
\end{equation}
It follows that, non perturbatively, $\beta_{\tilde g_1}$ admits two fixed points:
   \begin{equation} \label{eq:FP}
   \begin{split}
 \tilde g_{1\pm}&  =  \frac{
    \beta_1^{\tilde g} 
    \pm \sqrt{ (\beta_1^{\tilde g})^2 -\beta_0^{\tilde g}\beta_2^{ \tilde g} }
    }{\beta_2^{\tilde g}}  = \pm\sqrt{-\tilde g^2} +  \mathcal{O}(\tilde g^2) 
    \; , \crcr
  \beta'_{g_1}( \tilde g_{1\pm}) & = \pm 2\sqrt{ (\beta_1^{\tilde g})^2 -\beta_0^{\tilde g}\beta_2^{\tilde g} }
     = \pm \sqrt{-\tilde g^2} \left( 4\frac{\Gamma(\frac{d}{4})^2}{ \Gamma(\frac{d}{2})}\right) +  \mathcal{O}(\tilde g^3)\;.
   \end{split}
 \end{equation}
 The same holds for $\tilde g_2$ substituting $\tilde g$ with $\sqrt{q-1} \tilde g$, consequently we obtain four lines of fixed points parameterized by the marginal coupling $\tilde g$.
 
 \paragraph{\emph{Stability.}} The critical exponents are purely imaginary 
 for $\tilde g$ real and not too large, that is the fixed points are limit cycles (see Appendix \ref{app:ren}) and no trajectory can reach them.
 
 The situation, depicted in Fig.~\ref{fig:flows}, is much more interesting for a purely imaginary tetrahedral coupling $\tilde g = \pm\im |\tilde g|$. In this case the fixed point values of the pillow and double trace couplings are real (for $\tilde g$ not too large) and the critical exponents are also real. In particular $g_{1+}>0$ and $\beta'_{g_1}( g_{1+}) >0$, that is $(g_{1+},g_{2+}) $ is an infrared attractive fixed point (both the pillow and the double trace couplings are irrelevant). 

  \begin{figure}[htb]
\begin{center}
\includegraphics[scale=0.4]{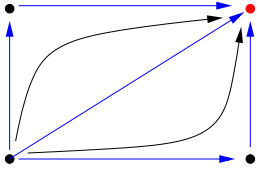} 
 \caption{Flows in the $(g_1,g_2)$ plane in the case of imaginary tetrahedral coupling.} 
 \label{fig:flows}
 \end{center}
\end{figure}
 
The tetrahedral invariant does not have any positivity property, but the pillow and double-trace do. It turns out that at the infrared fixed point $(g_{1+},g_{2+}) $, the real part of the action is bounded from below. 
 
 \paragraph{\emph{Comparison with $\zeta =1, d=4-\epsilon$.}}
 The case $\tilde g$ real is very similar to the Wilson-Fisher like fixed point discussed in Eq.~\eqref{eq:WFKlebanov}. In $d=4-\epsilon$ the tetrahedral coupling is not marginal (see Eq.~\eqref{eq:WFKlebanov}), but has a flow driven by the wave function. The flow has a fixed point for a real value of the tetrahedral coupling. 
 
 The key point is that for $\zeta =d/4$ the tetrahedral coupling is genuinely marginal and we are free to chose it purely imaginary. This leads to a stable infrared fixed point. 
 
\subsection{The infrared fixed point CFT}

The infrared fixed point $(g_{1+},g_{2+}) $ should be described by a melonic CFT. We will therefore attempt to identify the scaling dimensions and 
OPE coefficients of this CFT along the lines discussed in Section \ref{sec:lect0}. We consider the four point function:
\begin{equation}
\begin{split}
&  \frac{1}{N^6} \big\langle \phi_A(x_1) \phi_A(x_2)  \phi_B(x_3) \phi_B(x_4)  \big\rangle   \crcr
& = \frac{1}{N^3}  \big\langle \phi_A(x_1) \phi_A(x_2) \big\rangle 
\frac{1}{N^3} \big\langle \phi_B(x_3) \phi_B(x_4)  \big\rangle 
+ \frac{1}{N^6}\big\langle \phi_A(x_1) \phi_A(x_2)  \phi_B(x_3) \phi_B(x_4)  \big\rangle _{12 \to 34} \;.
\end{split}
\end{equation}
From Eq.~\eqref{eq:truc1}, the second term writes in terms of the four point kernel:
\begin{equation}
\frac{1}{N^6}\int_{y_1y_2}\left( \frac{1}{1-K} \right)_{AA; BB } (x_1,x_2 ; y_1,y_2)
   \left( G_{y_1 x_3}  G_{y_2x_4}  + G_{y_1x_4}  G_{y_2x_3}   \right) \;, 
\end{equation}
and from Eq.~\eqref{eq:onsheellker}, the four point kernel at leading order in $1/N$ is $K = K_1 P_1 + K_2 P_2$ with
$P_1 = 3 \left( \frac{1}{N^2} \delta^p - \frac{1}{N^3} \delta^d\right)$, $P_2= \frac{1}{N^3} \delta^d$ and:   
 \begin{equation}
 \begin{split}
&  K_1(x_1,x_2;y_1,y_2)  = G_{x_1z_1} G_{x_2z_2} \bigg[
\lambda^2 G_{z_1z_2}^{q-2} - \lambda_1 \delta_{z_1z_2} \bigg]  
     \delta_{z_1y_1} \delta_{z_2y_2} \;, \crcr
&  K_2(x_1,x_2;y_1,y_2)  = G_{x_1z_1} G_{x_2z_2} \bigg[ (q-1) \lambda^2 G_{z_1z_2}^{q-2} - \lambda_2 \delta_{z_1z_2} \bigg] 
     \delta_{z_1y_1} \delta_{z_2y_2} \;.
      \end{split}
 \end{equation}
Taking into account that $(P_1)_{AA;BB}=0$ and $(P_2)_{AA,BB} = N^3$ we obtain that only the term proportional to $P_2$ contributes: 
\begin{equation}
 \frac{1}{N^3}\int_{y_1y_2}\left( \frac{1}{1-K_2} \right) (x_1,x_2 ; y_1,y_2)
   \left( G_{y_1 x_3}  G_{y_2x_4}  + G_{y_1x_4}  G_{y_2x_3}   \right) \; .
\end{equation}
Recalling that the two point function in direct space is:
\begin{equation}
  G_{xy} = \frac{\Gamma(\Delta_{\phi})}{ 2^{ d-2\Delta_{\phi}} \pi^{d/2} \Gamma(\frac{d}{2} - \Delta_{\phi} ) Z }  \; \frac{1}{|x-y|^{2\Delta_{\phi}}}
   \;,
\end{equation}
we obtain that the eigenvalues of the kernel $k(\Delta, J)$ (see Section \ref{sec:lect0}) are:
\begin{equation}
 k(\Delta,J) =   (q-1) \tilde g^2 \;\;\Gamma(\Delta_{\phi})^{q-2}    
 \Gamma\left(  \frac{d}{2} - \Delta_{\phi}\right)^2
    \frac{
      \Gamma\left(\Delta_{\phi} - \frac{d}{2} + \frac{ \Delta + J}{2}\right)       \Gamma\left( \Delta_{\phi} - \frac{ \Delta - J}{2}\right)  }
   { 
   \Gamma\left(  d -   \Delta_{\phi}  - \frac{   \Delta - J}{2} \right) 
   \Gamma\left( \frac{d}{2} - \Delta_{\phi} +  \frac{  \Delta  + J}{2} \right) } \;,
\end{equation}
with $\Delta_{\phi}=d/q$.
The dimensions of the primary operators as well as the OPE coefficients can be computed starting from this formula. This study has been started in \cite{Benedetti:2019eyl} and the dimensions of the spin zero primaries have been analyzed. Surprisingly, for an imaginary (not too large) tetrahedral coupling one finds only real dimensions, while for a real tetrahedral coupling one finds complex dimensions of the form $d/2 + \im \alpha$.  

An interesting question at this stage is whether this large $N$ CFT is unitary.  In order to answer this question one needs to check whether the leading order OPE coefficients are real.
Pursuing this line of inquiry is a very interesting direction of research.

 \newpage

 \appendix

\section{The degree}\label{app:deg}
\renewcommand\theequation{\Alph{section}.\arabic{equation}}
\setcounter{equation}{0}

In this appendix we review the degree of edge colored graphs \cite{color,review} and reproduce the results cited in the main body of this paper. All these  results can be found in the literature.

We start by recalling some facts about ribbon graphs.
Ribbon graphs can be defined formally as combinatorial maps with an additional sign associated to the edges \cite{FabBolRiord} or as graphs embedded in surfaces. 
Being embedded they have vertices, edges and two dimensional cells which we call faces. 
A ribbon graph  $G$ can always be  projected onto the plane (see Fig.~\ref{fig:Ex-ribbon}). The projection of $G$ consists in:
 \begin{itemize}
  \item $V(G)$ \emph{ribbon vertices}.
  \item $E(G)$ \emph{ribbon edges} whose sides we call \emph{strands}. The edges can be straight (parallel strands) or twisted (twisted strands) and can cross. There is at most one twist per edge.
  \item $F(G)$ \emph{faces}, that is closed strands.
 \end{itemize} 
 The projection onto the plane is not canonical: by flipping the orientation on some of the vertices, some edges acquire twists and some twists are straightened out. 

\begin{figure}[H]
 \begin{center}
  \includegraphics[width=0.5\textwidth]{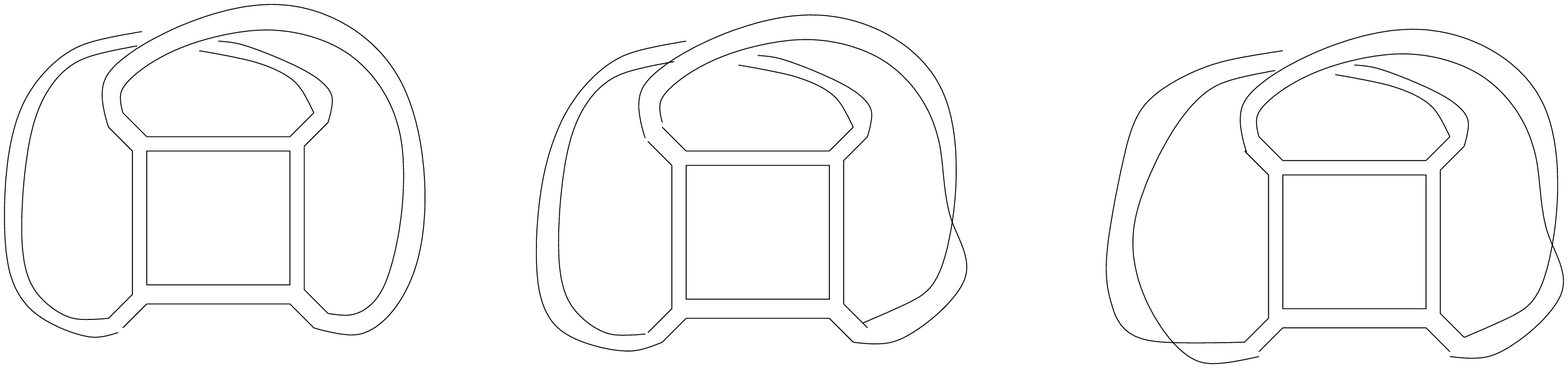}
  \caption{Examples of ribbon graphs with 4 vertices, 6 edges and, from left to right,  $2$, $2$ and $3$ faces. From left to right they are embedded into the torus, the Klein bottle and the real projective plane  $\mathbb{R}P^2$.}  \label{fig:Ex-ribbon}
 \end{center}
 \end{figure} 
 
 The Euler characteristic of a connected ribbon graph $G$ is $   V(G)-E(G)+F(G)  =  2-k(G)$ with $k(G)$ the \emph{non orientable genus} of $G$. The genus of a disconnected graph is the sum of the genera of its connected components. A connected ribbon graph with non orientable genus $k$ is embedded\footnote{To be precise, it is embedded in a surface with non orientable genus at least $k$ and the surface is orientable or not depending on whether the graph is orientable or not.} in a surface with non orientable genus $k$, that is:
 \begin{itemize}
  \item if $k=0$ then the graph is \emph{planar} and is embedded in the sphere.
  \item if $k$ is odd then the graph can only be embedded in a non orientable surface. Any projection onto the plane will have crossings and twists (see Fig.~\ref{fig:Ex-ribbon}, the rightmost case).
  \item if $k$ is even and non zero, then either the graph is:
     \begin{itemize}
      \item orientable, that is embedded in an orientable surface of genus $k/2$. It can be projected onto the plane with only crossing, but no twists (see Fig.~\ref{fig:Ex-ribbon} leftmost case).
      \item non orientable, that is embedded in a non orientable surface of non orientable genus $k$ (see Fig.~\ref{fig:Ex-ribbon}, the middle case). Any projection onto the plane   will required both crossings and twists.
     \end{itemize}
  \end{itemize}
   
 The edges in a ribbon graph ca be \emph{deleted} (see Fig.~\ref{fig:Del-edge}. The deletion of a ribbon edge consists in cutting the edge and joining together the strands at each end of the edge. 
 
\begin{figure}[htb]
 \begin{center}
  \includegraphics[width=0.5\textwidth]{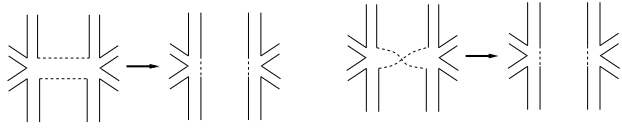}
  \caption{Deletion of an edge.}  \label{fig:Del-edge}
 \end{center}
 \end{figure}
 
 Let us delete iteratively a maximal set of edges in a connected graph such that at each step the edge we delete separates two different faces. This can not disconnect the graph. The number of edges deleted is $F(G)-1$. The remaining edges connect all the vertices, hence there are at least $V(G) -1$ of them. It follows that the non orientable genus of a connected ribbon graph is a non negative integer.

 \begin{proposition}\label{prop:genus}
 Consider a connected ribbon graph $G$ and denote $G'$ the graph obtained by deleting an edge $e$. Then: 
  \begin{itemize}
   \item either $G'$ consists in two connected components $G'_{1}$ and $G'_{2}$. In this case the genus is distributed between the connected components:
   $k(G) = k(G'_{1})+k(G'_{2})$.
   \item or $G'$ is connected an in this case the genus can not increase:
      $k(G') \le k(G)$.
  \end{itemize}
   \end{proposition}
\begin{proof}
  In the first case $E(G') = E(G)-1, V(G')=V(G)$, $F(G') = F(G) + 1$ and the vertices, edges and faces are distributed between the connected components of $G'$. Then
  \[
   4- k(G'_{1}) -k(G'_{2}) = V(G') - E(G') + F(G')  =2 +  2 - k(G) \;.
  \]
   In the second case $E(G') = E(G)-1, V(G')=V(G)$ and $F(G') \ge F(G) -1 $, hence: 
   \[
    2-k(G') = V(G') - E(G') + F(G') \ge 2-k(G) \;.
   \]
   
\end{proof}  

\begin{figure}[htb]
 \begin{center}
  \includegraphics[width=0.15\textwidth]{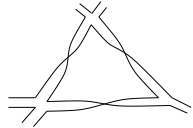}
  \caption{A triangle with twisted edges in a ribbon graph.}  \label{fig:triangle}
 \end{center}
 \end{figure}

\begin{proposition}\label{prop:twists}
A \emph{triangle} in a ribbon graph is a cycle of exactly three edges. If a connected ribbon graph $G$ contains a triangle of twisted edges (see Fig.~\ref{fig:triangle})
then $k(G)\ge 1$. 
\end{proposition}
\begin{proof} 
 We delete one by one all the edges incident to the triangle. In the process the graph $G$ splits into several connected components $G_{\rho}$. Let us denote $G_{1}$ the connected component consisting in the triangle. It has $3$ (bi-valent) vertices, $3$ edges and only $1$ face, hence $k(G_1) = 1$. Under the deletions the genus either decreases or is distributed between connected components, thus: 
 \[
  k(G) \ge \sum_{\rho\ge 1} k(G_{\rho}) \ge k(G_1) =1 \;.
 \]

\end{proof}

\subsection{The degree of edge colored graphs}

Edge colored graphs have been extensively discussed in detail in the literature \cite{color,RTM,review}. We review here some of their properties. 

\begin{definition}\label{def:colgr}
 An \emph{edge $(D+1)$--colored graph} $\cG$ is a graph with $D+1$ valent vertices and whose
  edges have a color $0,1,\dots D$ such that all the edges incident at a vertex have different colors. 
  
  The \emph{faces} with colors $(ij)$ of $\cG$ are the alternating cycles formed by edges with colors $i$ and $j$. We denote $V(\cG)$ the number of vertices and $F(\cG)$ the number of faces of $\cG$. 
\end{definition}

Let us consider a connected edge $(D+1)$--colored graph $\cG$. We can project it onto the plane by ordering the edges $0,1,\dots D$ (or any other order) clockwise around the vertices. There are $D!$ cyclic permutations $\pi$ over the colors $(0, \dots D)$.
We call \emph{jackets} of $\cG$ the ribbon graphs indexed by $\pi$ obtained by keeping all the vertices and edges of $\cG$ but only the faces with colors $( \pi^p(0) \pi^{p+1}(0))$. In a ribbon graph representation in which all the edges turn clockwise (following $\pi$) around the vertices, all the edges are twisted. Each of these ribbon graphs has a non orientable genus $k(\pi)$. The \emph{reduced degree} (or simply the \emph{degree}) of $\cG$ is the non negative number:
\begin{equation}\label{eq:defdegree}
 \boxed{\hat \omega(\cG) = \frac{1}{2(D-1)!} \sum_{\pi} k(\pi)  \ge 0 \; .}
\end{equation}
The reduced degree of a disconnected graph is the sum of the reduced degrees of its connected components.

    \begin{proposition}\label{prop:facecount}
      The total number of faces of a connected edge $(D+1)$--colored graph $\cG$ is:
      \[
       F(\cG) = D + \frac{D(D-1)}{4} V(\cG) -\hat \omega(\cG) \; .
      \]
      In particular the reduced degree is a non negative half integer.
     \end{proposition}
     \begin{proof}
       Every face $(ij)$ appears in $2(D-1)!$ cycles: the $(ij\dots)$ cycles and the $(i \dots j)$ cycles. Denoting $F^{(ij)}(\cG)$ the number of faces with colors $(ij)$ of $\cG$ we have:
      \begin{align*}
      & V(\cG) - \frac{ (D+1)}{2}V(\cG) +  \sum_{p=0}^{D}F^{( \pi^p(0) \pi^{p+1}(0) )}(\cG) = 2-k(\pi) \crcr
      & \qquad \qquad\Rightarrow F(\cG) =  D + \frac{D(D-1)}{4}V(\cG) -\frac{1}{2(D-1)!} \sum_{\pi} k(\pi)
      \; .       
      \end{align*}

     \end{proof}
     
     The discussion so far applies for $(D+1)$--colored graphs with $D\ge 2$. For $D=2$ the edge colored graphs are trivalent ribbon graph and the degree is just the genus. The fundamental difference between the $D=2$ and $D\ge 3$ cases comes from the family of graphs of degree zero. For $D=2$ they are the (edge $3$--colored) planar graphs. For $D\ge 3$ they are melonic graphs.
 
      \begin{definition}\label{def:melons}
       The \emph{ring graph} made of an edge of color $i$ closing onto itself and having $D$ faces  (with colors $(ij)$ for $j\neq i$) is melonic. 
       
       All melonic graphs are obtained by inserting iteratively two vertices connected by $D$ parallel edges 
       arbitrarily on the edges of lower order melonic graphs.  
       Melonic graphs are always connected.
      \end{definition}
      
      This definition pertains to vacuum graphs. Cutting any edge in a melonic vacuum graph one obtains a melonic two point graph.  Melonic two point graphs are such that their one particle irreducible components factor into $D$ parallel two point functions \cite{GurSch}. 
      
        \begin{figure}[H]
        \begin{center}
    \includegraphics[width=0.8\textwidth]{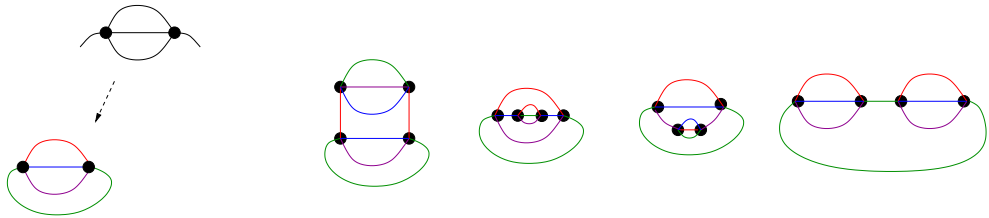}     
    \caption{Melonic graphs at first orders}         
        \end{center}
     \end{figure}
 
  \begin{proposition}\label{prop:melon} For $D\ge 3$, 
      a connected edge $(D+1)$--colored graph $\cG$ has reduced degree zero if and only if it is melonic.
     \end{proposition}
     \begin{proof}
As the insertion of two vertices connected by $D$ parallel edges brings $\binom{D}{2}$ new faces, it does not change the degree. It follows that melonic graphs have degree zero. 
      
      For the converse statement, we proceed by induction on the number of vertices. The faces are cycles with alternating colors hence have even length. We denote $F_{2p}(\cG)$ the number of faces of length $2p$ of $\cG$.
      A vertex contributes $\binom{D+1}{2}$ corners\footnote{A corner of a vertex is a couple of half edges $\{i,j\}$.} to the faces, thus 
      $\sum_{p\ge 0} 2p \, F_{2p}(\cG) = \frac{(D+1)D}{2}V(\cG) $. On the other hand 
      $\hat \omega(\cG)  = D + \frac{D(D-1)}{4}V(\cG) + \sum_{p\ge 1} F_{2p}(\cG)$, hence:
      \[
       \hat \omega(\cG) = D + \sum_{p\ge 1} F_{2p}(\cG)  \bigg( p \frac{D-1}{D+1} -1 \bigg) \;.
      \]
      As $D\ge 3$, the coefficient of $F_{2p}(\cG)$ is non negative for all $p\ge 2$. It follows that, if $\hat\omega(\cG)=0$ then $F_2 (\cG)> 0 $, that is the graph has at least a face of length exactly 2.
      
      Consider a face of length $2$ of $\cG$ formed by two edges of colors $i$ and $j$ which join two vertices $v$ and $w$. If $v$ and $w$ are joined by exactly one edge for all the colors the graph is melonic. Let a color $c$ such that two distinct edges of color $c$ are incident to $v$ and $w$. We call them $e^c_v$ and $e^c_{w}$. We consider the jacket $\pi = (icj\dots )$. This jacket is planar, $k( \pi )=0$. The two faces on the two sides of $e^c_v$ have colors $(i,c)$ and $(j,c)$ hence are different.
      Deleting $e^c_v$ we obtain a connected ribbon graph ${\cal J}_{\pi}'$ having one less face than $\pi$, and (using Proposition \ref{prop:genus}) $k({\cal J}_{\pi}') =k(\pi)=0$. We now delete $e^c_w$ to obtain ${\cal J}_{\pi}''$. This deletion increases the number of faces by $1$, $F({\cal J}_{\pi}'') = F({\cal J}_{\pi}')+1$. If ${\cal J}_{\pi}'' $
      were connected we would have $
       2- k({\cal J}_{\pi}'' ) = 4   \Rightarrow k({\cal J}_{\pi}'' ) = -2
       $
       which is impossible. Hence the deletion of $ e^c_v$ and $e^c_w$ disconnects the graph $\cG$. 
       
       Consider the graph obtained from $\cG$ by deleting $ e^c_v$ and $e^c_w$ and reconnecting the half edges directly in each connected component respecting the colors. It has the same numbers of edges and vertices as $\cG$, $D$ more faces (all the faces going through $e^c_v$ and $e^c_w$ are split) and two connected components $\cG_1$ and $\cG_2$, thus:
       \[
        \hat \omega(\cG_1) + \hat \omega(\cG_1) = 2D + \frac{D(D-1)}{4} V(\cG)- (F+D)
         = \hat \omega(\cG) =0,
       \]
       hence both $\cG_1$ and $\cG_2$ have degree zero and strictly fewer vertices than $\cG$. 
       
       Iterating we conclude that $\cG$ contains two vertices connected by $D$ parallel edges.
      
     \end{proof}
 
 \subsection{The SYK degree}
 
    Let $D\ge 3$ and denote
     $\cG^0$ the (possibly disconnected) edge $D$--colored graph obtained from a connected edge $(D+1)$--colored graph $\cG$ by erasing the edges of color $0$\footnote{For $D=3$, $\cG^0$ is a 3--colored graph, hence a trivalent ribbon graph.}. Being an edge colored graph, $\cG^0$ has a reduced degree $\hat \omega(\cG^0)$
    (defined as in Eq.~\eqref{eq:defdegree}, but with $D$ shifted to $D-1$). 
    
    \begin{proposition}\label{prop:SYK}
       The SYK degree of a connected edge $(D+1)$--colored graph $\cG$ 
       \begin{equation}
        \boxed{\Omega(\cG)  = \hat \omega(\cG) - \hat \omega(\cG^0) \;,} 
       \end{equation}
        is a half integer which obeys the bounds:
       \begin{equation}
          \frac{1}{D} \; \hat \omega(\cG) \le   \Omega(\cG)    \le \hat \omega(\cG) \;.
       \end{equation}
       The SYK degree is non negative and it is zero if and only if $\cG$ is melonic. 
       \end{proposition}
        \begin{proof}
            In order to prove the bounds,
            we observe that $\cG$ has $D!$ jackets and $\cG^0$ has $(D-1)!$ jackets. There is a $D$ to one correspondence between the jackets of $\cG$ and those of $\cG^0$ obtained by erasing the color $0$, that is $\pi = (0i\dots j) \to (i\dots j) = \pi^0 $.
            In the associated jacket of $\cG$ this corresponds to deleting the edges of color $0$. Observe that the graph corresponding to $\pi^0$ (which is a jacket of 
            $\cG^0$) might be disconnected. The genus can not decrease with the deletions, hence 
            $ \sum_{\pi}k(\pi) \ge D \sum_{\pi^0} k(\pi^0) $. We rewrite the SYK degree as:
            \begin{align*}
             \Omega(\cG) & = \frac{1}{2(D-1)!} \sum_{\pi} k(\pi) - 
             \frac{1}{2(D-2)!} \sum_{\pi^0} k(\pi^0) \crcr
             & = \frac{1}{2D (D-2)!} \left[  \sum_{\pi} k(\pi) - D \sum_{\pi^0} k(\pi^0) \right] + \frac{1}{2D!} \sum_{\pi} k(\pi) \;.
            \end{align*}
            
            The last statements follows from the bounds. 
            
            \end{proof}

   If $\cG$ is a melonic graph, then $\cG^0$ is 
   a union of melonic graphs. If $\cG^0$ happens to have only one connected component, then $\cG$ can be uniquely reconstructed from it: in the iterative construction of $\cG^0$ one ads an edge of color zero between the pair of vertices inserted at each step. 
    
 \subsection{The CTKT degree}

The graphs of the CTKT model are made of four valent stranded vertices connected by edges with three strands as depicted in  Fig.~\ref{fig:appCTKTvertices}. 
The faces are the closed strands and have a color. From left to right in Fig.~\ref{fig:appCTKTvertices} the vertices are the \emph{tetrahedral}, the \emph{pillow} and the \emph{double} trace vertex. There are three kinds of pillow vertices distinguished by the special color which is transmitted from on pair of half edges to the other. 
      \begin{figure}[htb]
        \begin{center}
    \includegraphics[width=0.4\textwidth]{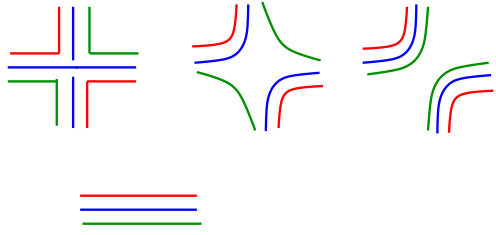}     
    \caption{Vertices and edges of the CTKT model}  \label{fig:appCTKTvertices}       
        \end{center}
     \end{figure}
We denote $V_t(\cG)$, $V_p(\cG)$ and $V_d(\cG)$ the numbers of tetrahedral, pillow and double trace vertices and $F(\cG)$ the number of faces of a graph $\cG$.

We aim to define jacket ribbon graphs which will allow us to count the faces. We can not do this naively due to the pillow and double trace vertices. So we first get rid of them. The pillow and double trace vertices can be resolved in terms of minimal configurations of the tetrahedral vertex. This is depicted in 
Fig.~\ref{fig:appresolve}. 

\begin{figure}[htb]
        \begin{center}
    \includegraphics[width=0.4\textwidth]{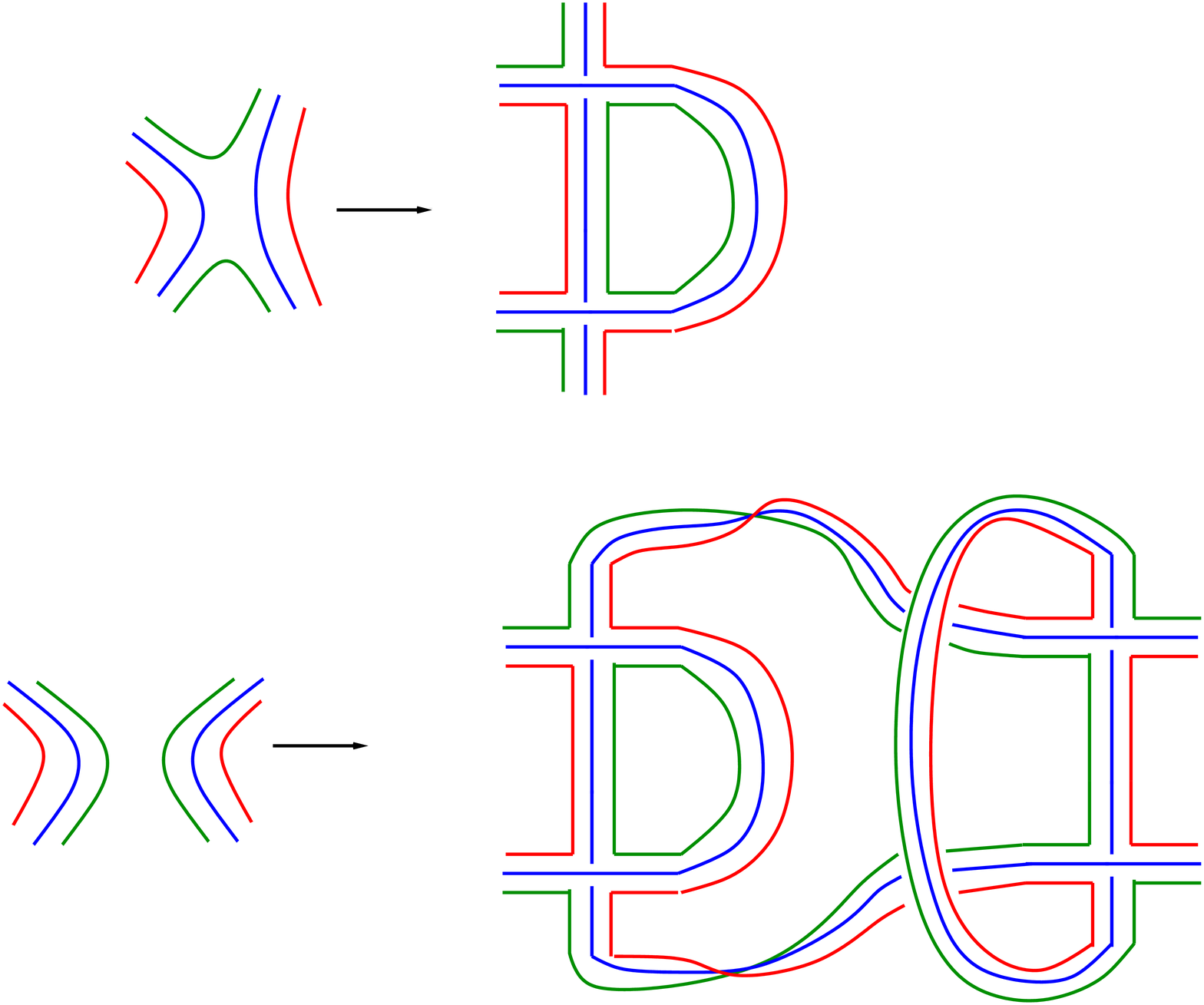}     
    \caption{Resolution of the pillow and double trace vertices in terms of the tetrahedral vertex}  \label{fig:appresolve}       
        \end{center}
\end{figure}

     For any graph $\cG$ we denote $\tilde \cG$ the graph obtained by replacing all the pillow and double trace vertices by their minimal resolutions in terms of tetrahedral vertices. We call this the \emph{refinement} of $\cG$. The refined graph $\tilde \cG$ has only tetrahedral vertices and:
     \[
      V_t(\tilde \cG) = V_t(\cG) + 2V_p(\cG) + 4V_d(\cG) \;,\qquad
      F(\tilde \cG) = F(\cG) + V_p(\cG) + 3V_d (\cG) \;.
     \]
     
     The refined graph $\tilde \cG$ admits three jacket ribbon graph ${\cal J}^i$ obtained by erasing the faces of the color $i$. We denote their non orientable genera $k({\cal J}^i)$. We define the \emph{CTKT degree} of $\cG$ (and of its refinement $\tilde \cG$) as:
     \begin{equation}\label{eq:degCTKT}
      \boxed{  \omega(\cG) = \frac{1}{2} \sum_{i} k({\cal J}^i) \ge 0 \;. }
     \end{equation}
    
    \begin{proposition}\label{prop:CTKTcount}
     The number of faces of a CTKT graph is:
     \[
      F(\cG) =   3 + \frac{3}{2} V_t(\cG) + 2V_p(\cG) + 3V_d(\cG) - 
               \omega(\cG) \;.
     \]
    \end{proposition}
     \begin{proof}
      Counting the faces of the refined graph $\tilde \cG$ by jacket we find 
      $F({\cal J}^i) = 2 - 2V(\tilde \cG) - k({\cal J}^i)$ hence
      $  F(\tilde \cG) = 3 + \frac{3}{2} V_t(\tilde \cG) -    \omega(\cG)$. 
      Expressing everything in terms of the numbers of vertices and faces of $\cG$ we find:
      \[
             F(\cG) =  3 + \frac{3}{2} V_t(\tilde \cG) - V_p(\cG) - 3V_d (\cG) -    \omega(\cG) = 3 + \frac{3}{2} V_t(\cG) + 2V_p(\cG) + 3V_d(\cG) - 
             \omega(\cG) \;.
      \]
      
     \end{proof}

     The CTKT degree is a half integer.
     The graphs of degree zero are a slight generalization of the melonic graphs.
     
     \begin{definition}
      We call a connected CTKT graph $\cG$ a \emph{melon-tadpole} graph if its refinement $\tilde \cG$ is a melonic graph.
      Such a graph is obtained by iterated insertions of melons and tadpoles into melons and tadpoles such that all the tadpoles are based at pillow or double trace vertices and all the melons have pairs of tetrahedral vertices. An example is presented in Fig.~\ref{fig:melon-tadpole}
     \end{definition}

\begin{figure}[htb]
        \begin{center}
    \includegraphics[width=0.2\textwidth]{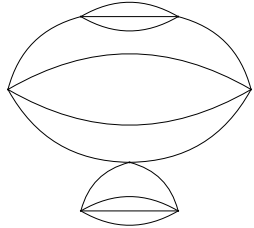}     
    \caption{A melon-tadpole graph.}  \label{fig:melon-tadpole}       
        \end{center}
\end{figure}

 \begin{proposition}\label{prop:CTKTlead}
  A $CTKT$ graph has reduced degree zero if and only if it is a melon-tadpole graph.
 \end{proposition}
   \begin{proof} As $  \omega(\cG) =   \omega(\tilde \cG)$, this comes to proving that $\tilde \cG$ has zero degree if and only if it is melonic. 
   The proof follows the one of Proposition~\ref{prop:melon}, but with some twists. The main difference is that faces can now have odd  length.
   
   Assume that the connected graph $\tilde \cG$ with only tetrahedral vertices has zero degree. Denoting $F_q(\tilde \cG)$ the number of faces of length $q$ of 
   $\tilde\cG$ and counting corners we have $\sum_{q\ge 1} qF_q(\tilde
   \cG) =  6 V(\tilde \cG)$. On the other hand $  \omega(\tilde \cG) = 
   3 + \frac{3}{2}V(\tilde \cG) - \sum_{q\ge 1} F_q(\tilde \cG) $, hence we get:
   \[
     \omega(\tilde \cG) = 3 + \sum_{q\ge 1} \left( \frac{q}{4} - 1  \right) F_q(\tilde \cG) \;.
   \]
   
   As the faces can now have odd length, $q=1,2$ and $3$ would have negative coefficients in the above formula. However, we have the following intermediate result.
 \begin{lemma}
    If a connected CTKT graph $\tilde \cG$ with only tetrahedral vertices has reduce degree zero, then $F_1(\cG) =F_3 (\cG)=0$.
   \end{lemma}
     \begin{proof}
      If $\tilde \cG$ has a face of length $1$ then it has a tadpole. We build the graph $\tilde \cG'$ by replacing the tadpole by an edge. This reduced both the number of edges and faces by 1, hence 
      $   \omega(\tilde \cG) =   \omega(\tilde \cG') 
      + 1/2 \ge 1/2 $.
      
      Now assume $\tilde \cG$ has a face of length $3$ and no tadpoles. Since 
      $\tilde\cG$ can not have a tadpole, then the face of length $3$ forms a triangle (it can not be a tadpole at the end of a dipole). In the jacket not containing the face of length 3 this leads to a triangle of twisted edges. From Proposition~\ref{prop:twists}, this jacket can not be planar.
      
   \end{proof}
   
   Thus $\tilde \cG$ must have a face of length $2$ and we conclude by the same induction as in Proposition~\ref{prop:melon}.
   
    \end{proof}

\newpage

\section{The renormalization (semi--)group}\label{app:ren}
\setcounter{equation}{0}

We briefly review the Wilsonian renormalization group \cite{peskin2018introduction} and use this opportunity to introduce some notation.

\paragraph{The one particle irreducible effective action.} The generating functional of connected moments of a theory with action $ \pmb{S}[\phi]$ is:
\[
 e^{\pmb{W}[J]} = \int [d\phi] \; e^{- \pmb{S}[\phi] + J\cdot \phi} \;,
 \qquad \frac{\delta \pmb{W}}{ \delta J_x} =  \Braket{\phi_x}^J \equiv \pmb{\phi}_x[J] \;, \qquad \frac{\delta^2 \pmb{W}}{ \delta J_x\delta J_y} = \Braket{\phi_x\phi_y}_c^{\pmb{J}}  = \pmb{G}_{xy}[J] \;,
\]
where this time we consider a local source $J_x$.
Going on shell means setting $J=0$. We denote $\pmb{J}[\phi]$ the solution of 
$\pmb{\phi}[J] = \phi$, that is $\pmb{J}[\phi]$ is the source that ensures that the expectation of the field is  exactly $\phi$. The Legendre transform of $\pmb{W}$ is:
\[
 \pmb{\Gamma}[\phi] =  \phi \cdot \pmb{J}  - \pmb{W}[\pmb{J}]
  \;,\qquad \frac{\delta \pmb{\Gamma}}{\delta \phi_x} = \pmb{J}_x
   \;,\qquad \frac{\delta^2 \pmb{\Gamma}}{\delta \phi_x\delta \phi_x} = 
    \pmb{G}^{-1}[\phi]_{xy} = 
    \bigg( \pmb{G}[J]^{-1}\big{|}_{J = \pmb{J}[\phi]} \bigg)_{xy} \; .
\]
Going on shell means setting  $\phi = \phi_{0}$ solution of the equations of motion
$\delta \pmb{\Gamma}/ \delta \phi  =0 $. From now on we consider that $\phi_0=0$, which can be guaranteed by taking an even action. The effective action can be written as a functional integral:
\[
 e^{-\pmb{\Gamma}[\phi]} = \int[d\psi] \;e^{-  \pmb{S}[\phi+\psi] + \pmb{J}[\phi] \cdot \psi }\; ,\qquad \Braket{\psi}^{\pmb{J}[\phi]} = 0 \;.
\]
In this form $\pmb{J}[\phi]$ is fixed by the requirement that, for the given background $\phi$, the expectation of $\psi$ is zero. The above functional integral can be evaluated in a Feynman expansion:
\[
 e^{-\pmb{\Gamma}[\phi]} =  e^{-\pmb{S}[\phi]}\int d\psi \; e^{ -\frac{1}{2} \psi \pmb{S}''[\phi]  \psi - \pmb{S}'[\phi] \psi - \sum_{n\ge 3} \frac{1}{n!} 
 \pmb{S}^{(n)}[\phi] \psi^n +  \pmb{J}[\phi] \psi  } \;.
\] 
We have $\pmb{\Gamma}[\phi] = \pmb{S}[\phi] + \pmb{\bar \Gamma}^{\rm 1PI}[\phi]$, where $- \pmb{\bar \Gamma}^{\rm 1PI}[\phi] $ is the sum over connected vacuum graphs of $\psi$ with propagator $( \pmb{S}''[\phi])^{-1}$,  vertices $ - \pmb{S}'[\phi]\psi $ and $-\pmb{S}^{(n)}[\phi]\psi^n, \; n\ge 3$ and a counterterm $\pmb{J}[\phi] \psi$
which ensures that $\braket{\psi}^{\pmb{J}[\phi]}=0$. Observe that the graphs contributing to $ \pmb{\bar \Gamma}^{\rm 1PI}[\phi] $ must have edges (i.e. the bare vertices are excluded).
 \begin{figure}[H]
 \begin{center}
 \psfrag{=0}{=0}
  \includegraphics[width=0.3\textwidth]{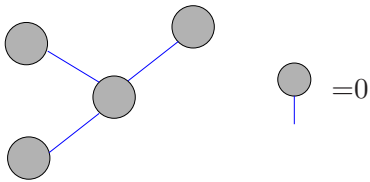}
  \caption{The 1PI decomposition of a graph.}\label{fig:1PIdecomp}
 \end{center}
 \end{figure} 
 These graphs can always be decomposed along one particle irreducibly edges (that is edges that, when cut, disconnect the graph), as depicted in 
 Fig.~\ref{fig:1PIdecomp}. Every graph is then a tree with vertices the one particle irreducible kernels. Due to the counterterm, the amplitude of any tree having a univalent leaf is zero, hence only the trees with exactly one vertex survive. The tree with one vertex is the sum over all the one particle irreducible graphs. This can be written formally as:
\[
 e^{-\pmb{\Gamma}[\phi]} = \int_{\rm 1PI} d\psi \; e^{ - \pmb{S}[\phi + \psi]} \;.
\]

There is a slight subtlety here, related to the bare vertices. Let us consider an action consisting in a free part, which is quadratic in the field, and an interaction:
\[
 \pmb{S}[\phi] = \frac{1}{2} \phi \cdot C^{-1} \cdot \phi  +
    \pmb{S}^{\rm int}[\phi] \;.
\]
The interaction $ \pmb{S}^{\rm int}[\phi]$ is conditioned to not have a term 
proportional to the free part, but it can have other quadratic terms, most notably a mass term. The functional $\pmb{\bar \Gamma}^{\rm 1PI}[\phi]$ contains only graphs with edges hence it does not contain bare vertices. One can include them in the ``full 1PI action''
$\pmb{  \Gamma}^{\rm 1PI}[\phi] =  \pmb{S}^{\rm int}[\phi] + \pmb{\bar \Gamma}^{\rm 1PI}[\phi]  $. The effective action writes:
\[
 \pmb{\Gamma}[\phi] = \pmb{S}[\phi]  + \pmb{\bar \Gamma}^{\rm 1PI}[\phi] =
    \frac{1}{2} \phi \cdot C^{-1} \cdot \phi  + \pmb{  \Gamma}^{\rm 1PI}[\phi]\;.
\]
We define the self energy, that is the amputated one particle irreducible 2 point function as:
\[
 \pmb{\Sigma}[\phi]_{xy} = - \frac{\delta^2  \pmb{\Gamma}^{\rm 1PI} }{ \delta \phi_x \delta\phi_y}  \;, \qquad \pmb{G}^{-1}[\phi]_{xy} = (C^{-1})_{xy} - \pmb{\Sigma}[\phi]_{xy} \;.
\]
The reason to include the bare vertices in the 1PI generating function is that, with this definition, the self energy includes the mass vertex (and also additional quadratic vertices, if they exist). Observe moreover that:
\[
 e^{-\pmb{\Gamma}[\phi]} = \int_{\rm 1PI} d\psi \; e^{ - \pmb{S}[\phi + \psi]}
  = e^{-\frac{1}{2} \phi \cdot C^{-1} \cdot \phi } \int_{\rm 1PI} d\psi \; 
  e^{ - \frac{1}{2} \psi\cdot C^{-1} \cdot \psi + \pmb{S}^{\rm int}[\phi + \psi]}
  \;.
\]

\paragraph{Renormalization group flow and fixed points.}
From now on we consider the free part of the action to be:
\[
  \frac{1}{2} \int d^d x \; \phi(x) (-\partial^2)^{\zeta} \phi(x) \;,
   \qquad C= \frac{1}{(-\partial^2)^{\zeta}} \;.
\]
with $\zeta\le 1$. The scaling dimension of the field 
$\phi(\Omega x) = \Omega^{-\Delta_{\phi}} \phi(x)$ is dictated by the free part to be $ \Delta_{\phi} = (d-2\zeta)/2$. Observe that the dimension of the field is $[\phi] = [ {\rm momentum} ]^{\Delta_{\phi}}$.

We introduce an ultraviolet cutoff $\Lambda$ and an infrared one $k$ and we replace $C$ by $C^{\Lambda}_k$, the covariance with 
cutoffs\footnote{A common choice is to use multiplicative momentum cutoffs $C^{\Lambda}_k(p) = C(p) \chi^{\Lambda}_k(p)$, where $\chi^{\Lambda}_k(p) = \Theta( p^2 / \Lambda^2 ) - \Theta(p^2/k^2)$, and $\Theta(u)$ is some approximated step function cutting off $u\ge 1$.}. In order to simplify the notation we sometimes suppress the UV cutoff, but one should remember that for now the UV cutoff is 
present\footnote{The cutoffs do not spoil the mass dimension of the field. We have:
\[
 C_k^{-1}(x-y) = \int_p \frac{p^{2\zeta}}{  \Theta( p^2 / \Lambda^2 ) - \Theta(p^2/k^2) } e^{-\im p(x-y)}  \;,
\]
hence $ C_{\Omega^{-1}k}^{-1}\big(\Omega(x-y) \big) = \Omega^{-d-2\zeta} C^{-1}_{k}(x-y)$ and the quadratic part is invariant under $x\to x' = \Omega x, k\to k ' = \Omega^{-1}k$.}
We parametrize the renormalization group flow by the effective action at 
scale $k$:
\begin{align}
  e^{-\pmb{\Gamma}_k[\phi]} = & e^{-\frac{1}{2} \phi \cdot (C_k)^{-1} \cdot \phi} \int_{\rm 1PI} [d\psi] \; e^{-\frac{1}{2} \psi \cdot(C_k)^{-1} \cdot \psi + \pmb{S}^{\rm int}[\phi + \psi] } \;, \crcr
  \label{eq:effaction}
  \pmb{\Gamma}_k[\phi] = & \frac{1}{n!} \sum_{n\ge 2} \int dx_1 \dots dx_n \;
   \Gamma^{(n)}_k(x_1,\dots x_n) \; \phi(x_1) \dots \phi(x_n) \; ,
 \end{align}
where $\Gamma_k^{(n\ge 3)}$ is the $n$-point amputated 1PI correlation and
$\Gamma_k^{(2)}$ is the inverse two point function:
\[ \Gamma^{(2)}_k= (G_k)^{-1} = (C_k)^{-1} - \Sigma_k \;, \]
with $G_k$ and $\Sigma_k $ the two point function and self energy 
with cutoffs. It follows that the inverse two point function in momentum space takes the form $ G_k^{-1} (p) = m_k + Z_k p^{2\zeta} + {\rm rest}$\footnote{The rest term is $p^2O(p^2/k^2)$ for $\zeta=1$, while for $\zeta<1$ it vanishes when lifting the cutoffs.} 
where $m_k=G_k^{-1}(0)$ is the renormalized mass parameter and $Z_k$ is the wave function renormalization.
The free part of the effective action $ \phi \cdot G_k^{-1} \cdot \phi$ is dimensionless hence the renormalized field $\sqrt{ Z_k}\phi$ has dimension $\Delta_{\phi}$\footnote{That is
$[ \sqrt{ Z_k} \phi ]= [{\rm momentum}]^{\Delta_{\phi}}$ and under a rescaling
of both the positions and the infrared cutoff we have
$
 \sqrt{ Z_{\Omega^{-1}k} } \phi( \Omega x) = \Omega^{-\Delta_{\phi} } \sqrt{ Z_k} \phi(x)
$.}. The anomalous dimension of the field is:
\[ \eta_k = -\frac{1}{2} k\partial_k \ln Z_k \;, \qquad \eta_k\sim O(g^2) \;.\]

It is customary to expand the interaction part of the effective action on a basis of local operators\footnote{This is done by Taylor expanding the fields in Eq.~\eqref{eq:effaction} around a position, say $x_1$:
\[\phi(x_i) = 
\sum_{q\ge 0} \frac{1}{q!} \sum_{\bar \mu = \mu_1 \dots \mu_q }(x_i -x_1)^{\bar \mu} \partial_{\bar\mu} \phi(x_1) \;.
\]
}:
\[
  \pmb{\Gamma}_k[\phi]= \frac{1}{2}  \phi \cdot G_k^{-1} \cdot \phi + 
  \sum_{n,J} 
   k^{d-J -n\Delta_{\phi}} Z_k^{n/2} g_k^{(n;J)}
  \int d^dx \; \partial^J \phi^n(x) \;, 
\]
where $   \partial^{J}  \phi^{n}(x)  $ denotes $J$ derivatives acting on $n$ fields in some order. The explicit $Z_k$ factors in the interaction make $ g_k^{(n;J)}$ dimensionless. The $\beta$ functions are the scale derivatives of the dimensionless couplings:
\[
 k\partial_k g_{k}^{i} = \beta^{i}(g) \;, \qquad
 \beta^{i}(g) =(\Delta_{i} - d) g_k^{i} + O(g^2) \;,
\]
where $i=(n,J)$ and  $\Delta_{i} =  J + n \Delta_{\phi} $ is the classical dimension of $ \partial^J \phi^n(x) $. We obtain a fixed point $(\eta_{\star}, g_{\star})$ if:
\begin{itemize}
 \item $\lim_{k\to 0} m_k = 0$, which requires to tune the bare mass in terms of the ultraviolet cutoff $\Lambda$.
 \item $\lim_{k\to 0} \eta_k = \eta_{\star}$ and $\lim_{k\to 0}\beta^i(g_{\star}) =0$.
\end{itemize}
Taking the ultraviolet cutoff to infinity and tuning the renormalized mass to zero, the only dimensionful parameter we are left with at the fixed point is $k$. We have $Z_k \overset{k\to 0}\sim k^{-2\eta_{\star}}$ and:
\[G_k(x-y) =  \frac{k^{2\Delta_{\phi}}}{Z_k} \, H(k|x-y|,g_{\star})
\overset{k\to 0}\sim  k^{2\Delta_{\phi} + 2\eta_{\star}} H(k|x-y|,g_{\star}) \;, \]
with $H$ some dimensionless function of the dimensionless argument $k|x-y|$ and the fixed point couplings $g_{\star}$. Taking $k\to 0$ we get the physical two point function:
\[
 G_k(x-y) \overset{k\to 0} \sim \frac{ c(g_{\star}) }{|x-y|^{2(\Delta_{\phi}+ \eta_{\star}) } } \;.
\]

In order to explore the neighborhood of the fixed point, let us denote 
$\nu_a$ and the eigenvalues of the stability matrix $\frac{\partial \beta_i}{\partial g_j} (g_{\star}) = P^{-1}_{ia} \nu_a P_{aj}$. At linear order in the perturbation $g_i = g_{i;\star}+ h_i(k)$ around the fixed point we have:
\[
  h_i(k) = P^{-1}_{ia}\left( \frac{k}{k_0}\right)^{\nu_a} P_{aj} h_j(k_0) \;,
\]
with $k_0$ the scale of the initial condition. Thus:
\begin{itemize}
 \item[-] if ${\rm Re}(\nu_a)>0$, the eigendirection is irrelevant
 (the perturbation vanishes for $k\to0$)
 \item[-] if ${\rm Re}(\nu_a)<0$, the eigendirection is relevant
 (the perturbation grows for $k\to0$)
 \item[-] if ${\rm Re}(\nu_a)=0$, then:
 \begin{itemize}
  \item if ${\rm Im}(\nu_a)=0$ the eigendirection is marginal and one gets a line of fixed points.
  \item if ${\rm Im}(\nu_a)=0$ the eigendirection is a limit cycle.
 This case is somewhat pathological because not only no trajectory can ever reach the fixed point, but also the exact value of the coupling at any given scale is strongly dependent on the initial condition. 
 \end{itemize}
\end{itemize}

The scaling dimensions of the operators are $d+\nu_a$.
In order to reach the fixed point one needs to fine tune the relevant couplings (the irrelevant ones flow by themselves to the fixed point values). A fixed point is predictive if it has a small number of relevant directions. 

The mass can be separated from the rest of the quadratic terms and treated as an interaction term. It is always classically relevant:
\[
 k\partial_k m_k = -2\zeta m_k + O(g^2) \;.
\]

\section{The wave function integral}\label{app:wf}
Let us compute for $\zeta  = d/2 - d/q \le 1$ the integral:
\[
I = \int_0^1 dt 
  \int_{0}^{\infty} d\alpha \; 
  \frac{ \prod_{i=1}^{q-1}  \alpha_i^{\zeta} }{
  \left( \sum_{i=1}^{q-1} \prod_{j\neq i} \alpha_j  \right)^{d/2+1}} 
 e^{ -t \frac{\prod_{i=1}^{q-1}  \alpha_i }{ \sum_{i=1}^{q-1} \prod_{j\neq i} \alpha_j  }  }  \;.
 \]
Changing variables to $\beta = \alpha^{-1}$, rescaling all the $\beta$s by $t$ and integrating out $t$ yields:
\begin{align*}
  I  = \frac{1}{\zeta} \int d\beta \; 
 \frac{\prod_{i=1}^{q-1} \beta_i^{d/2 - \zeta-1} }{ \left( \sum_{i=1}^{q-1} \beta_i\right)^{d/2+1}}   e^{-\frac{1}{\sum_{i=1}^{q-1} \beta_i}} \;.
\end{align*}
Introducing $x =  \sum_i \beta_i $ and $\beta_i =s_i x$ the integral becomes:
\begin{align*}
I= & \frac{1}{\zeta} \int_0^{\infty} d x
    \; x^{\zeta-2}  e^{-\frac{1}{x}}  \int_0^1 \frac{ds_1}{s_1}  s_1^{\frac{d}{2}-\zeta}
     \int_0^{1-s_1}  \frac{ds_2}{s_2} s_2^{\frac{d}{2} -\zeta } \dots \crcr
  & \qquad \qquad   \int_{0}^{1-s_1 -\dots -s_{q-3}} ds_{q-2} \;  
     s_{q-2}^{\frac{d}{2}-\zeta-1} (1 - s_1 - \dots s_{q-2})^{\frac{d}{2}-\zeta-1} \; ,
\end{align*}
and using:
\[
 \int_0^{1-x} ds \; s^{a-1} (1-x-s)^{b -1} =
 (1-x)^{a + b-1} \frac{\Gamma(a)\Gamma(b)}{\Gamma(a+b)} \;,
\]
we get:
\[
 I = \frac{\Gamma(1-\zeta) }{ \zeta } \; 
 \frac{ \Gamma\left( \frac{d}{2} -\zeta \right)^{q-1} }{  \Gamma\left[   \left( \frac{d}{2} -\zeta \right) (q-1)   \right]} 
  = \frac{\Gamma\left( 1 -\zeta \right) \Gamma\left( \frac{d}{2} -\zeta \right)^{q-1} }{ \zeta  \Gamma\left( \frac{d}{2} + \zeta\right)} \;.
\]

 \bibliographystyle{JHEP-3}
 
  \bibliography{/home/razvan/Desktop/lucru/Ongoing/Refs/Refs.bib}

\end{document}